\documentclass[twoside,leqno]{article}

\usepackage[letterpaper]{geometry}

\usepackage{siamproceedings}

\usepackage[T1]{fontenc}
\usepackage{amsfonts}
\usepackage{graphicx}

\Crefname{algocf}{Algorithm}{Algorithms}
\crefname{algocfline}{line}{lines}

\usepackage[ruled,vlined,linesnumbered,algonl,algo2e]{algorithm2e}
\SetEndCharOfAlgoLine{}
\SetKwComment{Comment}{\footnotesize$\triangleright$\ }{}

\SetCommentSty{mycommfont}

\usepackage{thmtools,thm-restate}

\newcommand{\NN}{\mathbb N}

\newcommand{\ZZ}{\mathbb Z}

\newcommand{\abs}[1]{\left\lvert #1 \right\rvert}

\newcommand{\SC}{\mathcal C}

\newcommand{\tst}{t^\star}
\newcommand{\dst}{d^\star}
\newcommand{\nst}{n^\star}

\newcommand{\eat}[1]{}

\usepackage{tikz}
\usetikzlibrary{arrows.meta}

\newcommand{\seq}[3]{#1^{(#2)}_{#3}}
\newcommand{\seqdef}[2]{\seq{#1}{#2}{}}

\begin{document}

    \title{\Large Language Generation in the Limit: Noise, Loss, and Feedback}
    \author{Yannan Bai\thanks{Department of Computer Science, Duke University
        (\email{yannan.bai@duke.edu}, \email{debmalya@cs.duke.edu}, \email{ian.zhang@duke.edu}).}
    \and Debmalya Panigrahi\footnotemark[1]
    \and Ian Zhang\footnotemark[1]
    }

    \date{}
    \maketitle



    \pagenumbering{arabic}
    \setcounter{page}{1}

    \begin{abstract}
        Kleinberg and Mullainathan (2024) recently proposed a formal framework
        called {\em language generation in the limit} and showed that given a sequence of example strings from an
        unknown target language drawn from any countable collection,
        an algorithm can correctly generate unseen strings from the target language within finite time.
        This notion of language generation was further refined by Li, Raman, and Tewari (2025),
        who defined progressively stricter categories called non-uniform and uniform generation
        within generation in the limit.
        They showed that a finite union of uniformly generatable collections is generatable in the limit,
        and asked if the same is true for non-uniform generation and generation in the limit.

        Our starting point in this paper is to resolve the question of Li, Raman, and Tewari in the negative:
        we give a uniformly generatable collection and a non-uniformly generatable collection
        whose union is not generatable in the limit.
        We then use facets of this construction to further our understanding
        of several variants of language generation.
        The first two, language generation with noise and without samples,
        were introduced by Raman and Raman (2025) and Li, Raman, and Tewari (2025) respectively.
        We show the equivalence of these models, for both uniform and non-uniform generation.
        We also provide a complete characterization of non-uniform noisy generation,
        complementing the corresponding result of Raman and Raman (2025) for uniform noisy generation.
        The former paper asked if there is any separation between noisy and non-noisy generation in the limit---we
        show that such a separation exists even with a {\em single} noisy string.
        Finally, we study the framework of generation with feedback, introduced by Charikar and Pabbaraju (2025),
        where the algorithm is strengthened by allowing it to ask membership queries.
        We draw a sharp distinction between finite and infinite queries:
        we show that the former gives no extra power, but the latter is closed under
        countable union, making it a strictly more powerful model than language generation without feedback.

        In summary, the results in this paper resolve the union-closedness of language generation in the limit,
        and leverage those techniques (and others) to give precise characterizations for natural variants
        that incorporate noise, loss, and feedback in language generation.
    \end{abstract}


    \section{Introduction.}\label{sec:intro}Buoyed by the tremendous popularity of large language models (LLMs), there has been a surge of recent interest in
formally understanding the phenomena of language learning and generation.
Suppose an algorithm is given an infinite sequence of distinct strings from an unknown language $K\in \SC$,
where $\SC$ is a collection of languages defined on a universe of strings $U$.
In this setting, the central question of {\em language generation in the limit} -- formulated recently by Kleinberg and
Mullainathan~\cite{KM24} -- is whether the algorithm can learn to correctly generate strings from $K$
within finite time.
The related problem of identifying the correct language $K$, rather than simply generating correctly from it,
was previously studied by several authors (Gold~\cite{Gol67}, Angluin~\cite{Ang80,Ang80a}),
but the results are largely negative: for almost any infinite collection $\SC$,
these results showed that the language identification problem is intractable.
It therefore came as a surprise when~\cite{KM24} showed the language {\em generation} problem was, in fact, tractable
    {\em for every countable collection of languages}.

This surprising positive result has led to a flurry of activity in the language generation problem.
Li, Raman, and Tewari~\cite{LRT25} classified collections of languages based on the finite time $\tst$
after which the algorithm correctly generates strings from the target language $K$.
If this time $\tst$ is independent of the target language $K$ and its enumeration,
then they called it {\em uniform generation};
if $\tst$ depends on $K$ but not its enumeration, they called it {\em non-uniform generation};
otherwise, the collection is simply said to be generatable in the limit.
They also derived structural characterizations for uniform
and non-uniform generation using the lens of classical learning theory.
Charikar and Pabbaraju~\cite{CP25} further studied the phenomenon of non-uniform generation in the limit,
and characterized general conditions and restrictions for this paradigm.
There have also been efforts at understanding variants of the basic model
that either strengthen or weaken the algorithm.
For instance, \cite{CP25} defined language generation with {\em feedback},
where the algorithm is strengthened by allowing it to ask membership queries of the adversary.
In contrast, Raman and Raman~\cite{RR25} defined a {\em noisy}
setting that yields a weaker algorithm, where the adversary is allowed to output
some strings that do not belong to the target language.
Another central theme has been understanding the fundamental tradeoff between {\em breadth} and {\em validity}
of language generation---the fraction of correct strings that the algorithm omits and wrong strings that it outputs.
Kalavasis, Mehrotra, and Velegkas~\cite{KMV25}
explored this tradeoff as captured by the complementary phenomena of mode collapse and hallucination
(see also Kalai and Vempala~\cite{KV24}).
This tension was also studied by Kleinberg and Wei~\cite{KW25} (see also \cite{KMV25a, CP25, PRR25}),
who established new definitions for the density of language generation
that help quantify breadth in language generation.

\subsection{Our Results.}

Inspired by these impressive developments in a short time frame, we consider a set of fundamental questions around
the power and limitations of language generation in this paper.
Our starting point is a question posed by~\cite{LRT25}.
They showed that any finite union of uniformly generatable collections is generatable in the limit,
and asked if the same is true for non-uniform generation and generation in the limit.
For the last category, this would imply closedness under finite union;
it is already known that the category is not closed under infinite (even countable) union.
Our first result refutes this in the strongest sense by defining just two collections
that are respectively non-uniformly and uniformly generatable (even without samples, i.e., without example strings from
the adversary), but their union is not generatable in the limit.

\begin{restatable}{theorem}{union}
    \label{thm:union_ungen}
    There exist collections $\SC_1$ and $\SC_2$ such that $\SC_1$ is non-uniformly generatable (without samples)
    and $\SC_2$ is uniformly generatable (without samples),
    but $\SC_1 \cup \SC_2$ is not generatable in the limit.
\end{restatable}

The above theorem is also interesting from a conceptual perspective since it further distinguishes the phenomenon of
language generation from traditional models of learning, a line of research started by~\cite{LRT25}.
Indeed, for traditional learning, multiple learners can be amalgamated into a combined learner, such as in boosting,
but the above theorem rules this out for language generators in general.

\medskip\noindent{\bf Lossy and Noisy Generation.}
Next, we consider two natural variants of the language generation model
that are weaker than the basic model of~\cite{KM24}.
A noisy model of language generation was introduced by~\cite{RR25} where the adversary is constrained to output all
correct strings from $K$ but can also output a finite number of incorrect strings that do not belong to $K$.
In a similar vein,~\cite{LRT25}
introduced language generation without samples (they called it auto-regressive generation),
where the adversary does not provide any example strings,
and the algorithm must generate entirely by itself.
Our next result establishes an equivalence between these two models, for both uniform and non-uniform generation.
(Note that in language generation without samples, non-uniform generation
is equivalent to generation in the limit since there is no enumeration.)

\begin{restatable}{theorem}{equiuniform}
    \label{thm:equi_lossy_noisy}
    A collection $\SC$ is uniformly noisily generatable if and only if
    $\SC$ is uniformly generatable without samples.
    Similarly, a collection $\SC$ is non-uniformly noisily generatable if and only if
    $\SC$ is generatable in the limit without samples (equivalently, non-uniformly generatable without samples).
\end{restatable}

We also complement the existing characterization for uniform noisy generation
in~\cite{RR25}\footnote{The result of~\cite{RR25} shows that a collection $\SC$ is uniformly generatable without samples
if and only if $\abs{\bigcap_{L \in \SC} L} = \infty$.}
by providing a complete characterization for non-uniform noisy generation.
By \cref{thm:equi_lossy_noisy},
this also means we now have complete characterizations for uniform and non-uniform generation without samples.

\begin{restatable}{theorem}{charnonuniform}
    \label{thm:gen_limit_without}
    A collection $\SC$ is generatable in the limit without samples
    if and only if there exists a countable sequence of collections $\SC_0 \subseteq \SC_1 \subseteq \dots$
    such that $\SC = \bigcup_{i \in \NN} \SC_i$ and $\abs{\bigcap_{L \in \SC_i} L} = \infty$ for all $i \in \NN$.
\end{restatable}

Our next contribution is to provide new definitions that offer quantitative refinements for both settings.
The first is a new concept of {\em lossy generation}
where the adversary can omit some (but possibly not all) strings from the target language $K$.
Clearly, the extreme case where the algorithm does not receive any input from the adversary at all is the previously
studied case of language generation without samples.
In the intermediate case, we allow the adversary to omit an infinite number of strings (while still enumerating an infinite subset of the target language).
In this model, we show that uniform and non-uniform generation is equivalent to uniform and non-uniform generation in the non-lossy setting.
\begin{theorem}
    If an algorithm $G$ uniformly generates for a collection $\SC$,
    then $G$ also uniformly generates for $\SC$ with infinite omissions.
    Similarly, if an algorithm $G$ non-uniformly generates for a collection $\SC$,
    then $G$ also non-uniformly generates for $\SC$ with infinite omissions.
\end{theorem}

We may also ask what happens if the adversary omits only a finite number of strings from the target language $K$.
Intuitively, this seems similar to the non-lossy setting since the target language is an infinite set.
Quite surprisingly, we show a strong separation between lossy and non-lossy generation---we
show that there are generatable collections that become ungeneratable
even if the algorithm omits just {\em one} string.

We take this fine-grained lens to noisy generation as well, and explore the impact of the level of noise
(quantified by the number of incorrect strings output by the adversary) on language generation.
\cite{RR25} asked if all generatable collections are also generatable with noise.
We answer this question in the negative in the strongest sense:
we show that there are generatable collections where even a {\em single}
incorrect string output by the adversary makes the collection ungeneratable.

Indeed, the same collection of languages gives us these two separations, which we state in the following theorem:
\begin{restatable}{theorem}{finegrained}
    \label{thm:fine_grained}
    For every $i \in \NN$, there exists a collection that is generatable in the limit with $i$
    omissions or with noise level $i$, but is not generatable in the limit with either $i+1$
    omissions or with noise level $i+1$.
\end{restatable}

This shows that, in particular for $i=0$, there is a collection that is generatable in the limit,
but becomes ungeneratable if a single string is omitted or if a single incorrect string is output by the adversary.
Furthermore, we show that knowledge of the level of noise is crucial to the language generation process.
In particular, we show a separation between the original model of generation in the limit with noise
due to~\cite{RR25} where the level of noise is finite but unknown to the algorithm,
and generation in the limit with noise level $i$, for every $i$.

\begin{restatable}{theorem}{noisesensitive}
    \label{thm:noise_sensitivity}
    There exists a collection that is generatable in the limit with noise level $i$ for any $i \in \NN$,
    but is not noisily generatable in the limit.
\end{restatable}

\medskip\noindent{\bf Generation and Identification with Feedback.}
We also consider the model introduced by~\cite{CP25} -- called language generation with {\em feedback}
-- where the algorithm can ask membership queries, i.e., whether a particular string is in the target language.
We precisely derive the role of feedback in language generation: we show that any collection that is generatable with
finite feedback is also generatable without feedback; in contrast, there are collections that are generatable with
infinite feedback but not without feedback.

\begin{restatable}{theorem}{feedbackcombined}
    \label{thm:feedback_combined}
    Language generation with infinite feedback is strictly more powerful than with finite feedback, the latter being
    equivalent to generation without feedback.
\end{restatable}

Finally, we consider the problem of language identification in the feedback model. In this problem, the algorithm must not only generate correct strings from the target language, but also correctly identify the index of the target language in a countable collection of languages. 
We show that any countable collection can be non-uniformly identified with infinite feedback.
\begin{restatable}{theorem}{identfeed}
    Any countable collection $\SC$ can be non-uniformly identified with feedback.
\end{restatable}

We give a figurative representation of the relationship between the different models of language generation that we consider in this paper in \Cref{fig:venn}.

\begin{figure}
    \centering
    \begin{tikzpicture}[thick,scale=0.8, every node/.style={scale=0.8}, font=\small]


\def\Rbig{5.2}      
\def\Rleft{3.00}    
\def\Rright{4.00}   
\def\Rinner{1.55}   
\def\Rnoisy{6.3}
\def\Rnoisya{7.4}
\def\Rlim{8.5}
\def\Rfeed{9.6}

\coordinate (O) at (0,0);
\coordinate (L) at (-1.75,-0); 
\coordinate (R) at (1,-0); 
\coordinate (I) at (-1,-0); 
\coordinate (N) at (0,0); 
\coordinate (Na) at (0,0); 
\coordinate (Lim) at (0,0); 
\coordinate (Feed) at (0,0);

\fill[gray!60] (O) circle (\Rfeed);

\begin{scope}[transparency group]
  \fill[white,opacity=.72] (L) circle (\Rleft);   
  \fill[white,opacity=.60] (R) circle (\Rright);  
  \fill[white,opacity=.85] (I) circle (\Rinner);  
  \fill[white,opacity=.50] (O) circle (\Rbig);
  \fill[white,opacity=.45] (N) circle (\Rnoisy);
  \fill[white,opacity=.35] (Na) circle (\Rnoisya);
  \fill[white,opacity=.25] (Lim) circle (\Rlim);
\end{scope}

\draw[gray!25] (L) circle (\Rleft);
\draw[gray!25] (R) circle (\Rright);
\draw[gray!25] (I) circle (\Rinner);
\draw[gray!25] (O) circle (\Rbig);
\draw[gray!25] (N) circle (\Rnoisy);
\draw[gray!25] (Na) circle (\Rnoisya);
\draw[gray!25] (Lim) circle (\Rlim);
\draw[gray!25] (Feed) circle (\Rfeed);

\node[align=center] at (0,4.5) {\bfseries Non-uniform};
\node[align=center] at (-3.85,-0) {\bfseries Uniform};
\node[align=center] at (3.12,-0) {\bfseries Non-unif. noisy = \\ \bfseries Non-unif. w/o samples};
\node[align=center] at (-1,-0) {\bfseries Unif. noisy = \\ \bfseries Unif. w/o samples};
\node[align=center] at (0,5.65) {\bfseries Noisy gen. in the limit};
\node[align=center] at (0,6.75) {\bfseries Gen. with noise level $i$};
\node[align=center] at (0,7.85) {\bfseries Gen. in the limit = \\
\bfseries Gen. in the limit with finite feedback};
\node[align=center] at (0,8.9) {\bfseries Gen. in the limit with infinite feedback};
\end{tikzpicture}

\caption{A figurative representation of the relationship between the various models of language generation that we consider in this paper. The equivalence of (non)-uniform noisy and (non)-uniform generation without samples is given by~\cref{thm:equi_lossy_noisy}.
The separations between noisy generation in the limit, genereration with noise level $i$, generation in the limit, and generation in the limit with infinite feedback are given by~\cref{thm:noise_sensitivity,thm:fine_grained,thm:feedback_combined}.}
\label{fig:venn}
\end{figure}

%
%
%
%
%
%

\medskip\noindent{\bf Concurrent and Subsequent Work.}
In concurrent and independent work, Hanneke, Karbasi, Mehrotra, and Velegkas~\cite{HKMV25}
also proved \Cref{thm:union_ungen}, thereby resolving the open question posed by \cite{LRT25}.
In work that appeared subsequent to our paper, Charikar and Pabbaraju~\cite{CP25b} explored 
an appropriately defined Pareto frontier for non-uniform generation.

\medskip\noindent{\bf Roadmap.}
We start with some preliminary definitions in \Cref{sec:prelim}.
The result on union-closedness (\Cref{thm:union_ungen}) appears in \Cref{sec:basic}.
The results on lossy and noisy generation appear in \Cref{sec:lossy,sec:noisy} respectively.
The results on generation with feedback appear in \Cref{sec:feedback}.

    \section{Preliminaries.}\label{sec:prelim}We follow the framework introduced by~\cite{KM24}, along with additional variations and definitions given by~
\cite{LRT25, RR25, CP25}.
A language $L$ is an infinite subset of a countably infinite set $U$ called the universe, and a collection $\SC$
is a (possibly uncountable) set of languages.
Unless specified otherwise, we will assume without loss of generality
that all collections are over the set of integers $\ZZ$.
We also denote the set of nonnegative integers by $\NN = \{0, 1, \dots\}$
and abbreviate contiguous elements of a sequence $x_i$, \dots, $x_j$ by $x_{i:j}$.

\subsection{Generation in the Limit.}
In the general setup, there is a fixed collection $\SC$ and a target language $K \in \SC$
selected by the adversary.
The adversary then presents the strings of $K$ in an enumeration $x_0$, $x_1$, \dots,
where each $x_t$ is contained in $K$, and for every $z \in K$, there exists some $t$ where $z = x_t$.
In addition, we require that every string in the enumeration is unique.
In previous literature, the adversary has been allowed to repeat strings in its enumeration.
However, we show in~\cref{sec:appendix} that generation is equivalent regardless of whether
the adversary is allowed to repeat strings in its enumeration.
Thus, we require every string in the enumeration to be unique, which simplifies some of the proofs and definitions.

At each time step $t$, the algorithm takes as input the set of strings enumerated by the adversary so far
and outputs a new string $z_t$. 
The goal is that after some finite time $\tst$, all strings
$z_t$ for $t \ge \tst$ are correct \emph{unseen} strings from the target language $K$.
Note that unlike the adversary, the algorithm is allowed to output the same string multiple times,
but the algorithm's string is only considered correct if it is distinct from all the example strings
given by the adversary so far.


For any enumeration $x$, we will use $S(x)_t = \{x_0, x_1, \dots, x_t\}$ to denote
the set of strings enumerated up until time $t$.
When the enumeration is clear from context, we will simply write $S_t$.
We also use $S_{\infty} = \bigcup_{i \in \NN} \{x_i\}$ to denote the entire set of enumerated strings.

\begin{definition}[Generator algorithm~\cite{LRT25}]
    A generator algorithm is a function $U^* \to U$ which takes as input a finite ordered
    set of strings $x_0$, \dots, $x_t$, and outputs a string $z_t$.
\end{definition}

\begin{definition}[Generation in the limit~\cite{KM24}]
    An algorithm $G$ generates in the limit for a collection $\SC$
    if for any $K \in \SC$ and any enumeration $x$ of $K$,
    there exists a time $\tst$ such that for all $t \ge \tst$,
    the generated string $z_t$ at time $t$ is in $K \setminus S_t$.
\end{definition}

Note that in the above definition, the time $\tst$ at which the algorithm must generate correctly
can be a function of both the target language $K$ and the adversary's enumeration $x$ of $K$.
A stricter requirement would be that $\tst$ is independent of the
enumeration or even the target language.
These notions were formalized by~\cite{LRT25} to define non-uniform and uniform generation.

\begin{definition}[Non-uniform generation~\cite{LRT25}]
    An algorithm $G$ non-uniformly generates for a collection $\SC$
    if for any $K \in \SC$, there exists a time step $\tst$
    such that for every enumeration $x$ of $K$ and every time $t \ge \tst$,
    the generated string $z_t$ is in $K \setminus S_t$.
\end{definition}

\begin{definition}[Uniform generation~\cite{LRT25}]
    An algorithm $G$ uniformly generates for a collection $\SC$
    if there exists a time step $\tst$ such that for any $K \in \SC$ and
    any enumeration $x$ of $K$,
    the generated string $z_t$ for every time $t \ge \tst$
    is in $K \setminus S_t$.
\end{definition}

The following are some useful properties and combinatorial dimensions of a collection of languages.
\begin{definition}[Consistent languages]
    Given a collection $\SC$, the set of consistent languages for a set $S$ is the set
    $\{L \in \SC \mid S \subseteq L\}$ of languages in $\SC$ that contain $S$.
\end{definition}

\begin{definition}[Closure~\cite{LRT25}]
    Given a collection $\SC$, the closure of a set $S$ is the intersection
    of all consistent languages for $S$ in $\SC$.
\end{definition}

\begin{definition}[Closure dimension~\cite{LRT25}]
    The closure dimension of a collection $\SC$ is the size of the largest set
    $S = \{x_1, \dots, x_d\}$ such that the closure of $S$ in $\SC$ is finite.
\end{definition}

\subsection{Noisy Generation.}
Raman and Raman~\cite{RR25} introduced a model of noisy generation where there may be
a finite number of extraneous strings in the adversary's enumeration.
As before, the adversary selects a language $K \in \SC$
and an enumeration $y_0$, $y_1$, \dots\ of $K$.
However, the adversary is now allowed to insert $\nst$ \emph{unique} strings not belonging to $K$ for any finite
\emph{noise level} $\nst \in \NN$ into the enumeration $y_0$, $y_1$, \dots\ to
obtain a noisy enumeration $x_0$, $x_1$, \dots.
The noisy enumeration is then presented to the algorithm $G$, which must eventually generate
correct unseen strings from the target language $K$.

\begin{definition}[Noisy enumeration]
    For any infinite language $K$, a noisy enumeration of $K$
    is any infinite sequence $x_0$, $x_1$, \dots\ without repetitions,
    such that $K \subseteq \bigcup_{i \in \NN} \{x_i\}$ and
    $\abs{\bigcup_{i \in \NN} \{x_i\} \setminus K} < \infty$.
\end{definition}

Similar to generation without noise, we have the corresponding definitions for generation with noise based on
how the time step $\tst$ at which we generate correctly is quantified.

\begin{definition}[Uniform noisy generation~\cite{RR25}]
    An algorithm $G$ uniformly noisily generates for a collection $\SC$
    if there exists a time $\tst$ such that for every $K \in \SC$ and every noisy enumeration $x$ of $K$,
    the algorithm's output $z_t$ is in $K \setminus S_t$ for all times $t \ge \tst$.
\end{definition}

\begin{definition}[Non-uniform noisy generation~\cite{RR25}]
    An algorithm $G$ non-uniformly noisily generates for a collection $\SC$
    if for every $K \in \SC$, there exists a time $\tst$ such that for every noisy enumeration $x$ of $K$,
    the algorithm's output $z_t$ is in $K \setminus S_t$ for all times $t \ge \tst$.
\end{definition}

\begin{definition}[Noisy generation in the limit~\cite{RR25}]
    An algorithm $G$ noisily generates in the limit for a collection $\SC$
    if for every $K \in \SC$ and every noisy enumeration $x$ of $K$, there exists a time $\tst$ such that
    the algorithm's output $z_t$ is in $K \setminus S_t$ for all times $t \ge \tst$.
\end{definition}

\cite{RR25} provides two additional variants of uniform and non-uniform noisy generation
where the time step $\tst$ can depend on the noise level, but we do not discuss those variations here.



\subsection{Projection.}
We define a new notion of projecting a collection onto a smaller universe,
which will be useful for reasoning about generation.

\begin{definition}[Projection of a language]
    Given any language $L$, the projection of $L$ onto a universe $U'$
    is defined as $L \cap U'$.
\end{definition}

\begin{definition}[Projection of a collection]
    Given any collection $\SC$ of languages, the projection of $\SC$ onto a universe $U'$
    is defined as $\{L \cap U' \mid L \in \SC, \abs{L \cap U'} = \infty\}$.
\end{definition}
We can imagine the projection of $\SC$ onto $U'$ as projecting each language $L \in \SC$
onto $U'$, and then removing the projections that have a finite number of elements.

    \section{Generation in the Limit is Not Closed under Finite Union.}\label{sec:basic}In their work defining uniform and non-uniform generation, Li, Raman, and Tewari~\cite{LRT25}
showed that the finite union of uniformly generatable collections is generatable in the limit.
They asked if the same holds for the broader classes, i.e.,
whether the finite union of non-uniformly generatable classes are generatable in the limit,
and whether generatability in the limit is closed under finite unions (Questions $6.2$, $6.3$).
In this section, we refute these possibilities in the strongest sense.
We give two collections $C_1$ and $C_2$ where $C_1$ is generatable in the limit
(even without samples, i.e., with no example strings provided by the adversary),
$C_2$ is uniformly generatable in the limit (without samples),
but $C_1 \cup C_2$ is not generatable in the limit.
Note that collections that are generatable in the limit without samples are also non-uniformly generatable.
Thus this negatively answers Questions $6.2$, $6.3$ in~\cite{LRT25}.

Recall that without loss of generality, the universe $U$ is the set of all integers $\ZZ$.
For each integer $i \in \ZZ$, define $P_i = \{i, i+1, i+2, \dots\}$
to be the infinite increasing sequence of integers starting at $i$.



\union*


\begin{proof}

    Let $\SC_1 = \bigcup_{i \in \NN} \{A \cup P_i \mid A \subseteq \ZZ\}$.
    A language in collection $\SC_1$ comprises an arbitrary subset of integers
    along with all integers starting from some value $i \in \NN$.
    The collection $\SC_1$ is generatable in the limit without samples
    since the algorithm can output $0$, $1$, $2$, \dots.
    The algorithm starts generating correctly from time $\tst = i$.

    Let $\SC_2 = \{A \cup \ZZ_{< 0} \mid A \subseteq \ZZ\}$.
    A language in collection $\SC_2$ comprises an arbitrary subset of integers along with all negative integers.
    The collection $\SC_2$ is uniformly generatable without samples since
    the algorithm can output $-1$, $-2$, $-3$, \dots.
    The algorithm only generates correct integers, i.e., starting from time $\tst = 0$.

    We now show that $\SC = \SC_1 \cup \SC_2$ is not generatable in the limit.
    Assume for contradiction that there exists an algorithm $G$ which generates in the limit for $\SC$.
    We will construct an enumeration $\{x_i\}_{i \in \NN}$ of a language $K \in \SC_2$ such that there
    exists an infinite sequence of times $t_0 < t_1 < \dots$ where the string
    output by $G$ at each time $t_i$ is not in $K$.
    This gives the desired contradiction.

    The adversary's enumeration proceeds in stages,
    where $L_j$ denotes the language that the adversary enumerates in stage $j$.
    This enumeration is denoted $\seqdef{x}{j}$,
    where $\seq{x}{j}{i}$ denotes the integer output by the algorithm from the language $L_j$ at time $i$.

    Initially, we are in stage $0$.
    In this stage, $L_0 = \NN$, and the adversary's enumeration is given by $\seq{x}{0}{i} = i$.
    Since $G$ generates in the limit, there must exist a time $t_0$
    where the algorithm outputs an integer $z_{t_0} > t_0$.
    For $i \le t_0 + 1$, we set
    \[x_i =
    \begin{cases}
        i & i \le t_0 \\
        -1 & i = t_0+1
    \end{cases}.\]
    In other words, the adversary follows the enumeration $\seqdef{x}{0}$
    till time $t_0$ and then generates $-1$ in the next timestep.

    We now enter stage $1$.
    Recall that $S_t$ denotes the set of enumerated strings $\{x_0, \dots, x_t\}$ until time $t$.
    Let $L_1 = S_{t_0+1} \cup \{z_{t_0}+2, z_{t_0}+3, \ldots\}$ be a language in $\SC_1$,
    and let the adversary's enumeration be given by
    \[\seq{x}{1}{i} =
    \begin{cases}
        \seq{x}{0}{i} & i \le t_0\\
        -1 & i = t_0 + 1\\
        z_{t_0} + (i-t_0) & i \ge  t_0 + 2
    \end{cases}.\]
    This definition ensures that the adversary's enumeration is valid,
    and that the algorithm's output at time $t_0$ is incorrect since $z_{t_0} \notin L_1$.
    Once again, since $G$ generates in the limit, there must exist a time $t_1 > t_0 + 1$
    where the algorithm outputs an integer $z_{t_1} > z_{t_0} + (t_1 - t_0)$.
    For $t_0 + 2 \le i \le t_1 + 1$, we set
    \[x_i =
    \begin{cases}
        z_{t_0} + (i-t_0) & t_0 + 2 \le i \le t_1 \\
        -2 & i = t_1+1
    \end{cases}.\]
    In other words, the adversary follows the enumeration $\seqdef{x}{1}$
    from time $t_0+2$ to time $t_1$ and then generates $-2$ in the next timestep.

    We proceed iteratively in this manner.
    Stage $j$ is entered at time $t_{j-1} + 2$.
    Let $L_j = S_{t_{j-1}+1} \cup \{z_{t_{j-1}}+2, z_{t_{j-1}}+3, \ldots\}$ be a language in $\SC_1$,
    and let the adversary's enumeration be given by
    \[\seq{x}{j}{i} =
    \begin{cases}
        \seq{x}{j-1}{i} & i \le t_{j-1}\\
        -j & i = t_{j-1} + 1\\
        z_{t_{j-1}} + (i-t_{j-1}) & i \ge  t_{j-1} + 2
    \end{cases}.\]
    As earlier, this definition ensures that the adversary's enumeration is always valid,
    and that the algorithm's output at time $t_{j'}$ for all $j' < j$ is incorrect
    since each $z_{t_{j'}}$ is not in $L_j$.
    Once again, since $G$ generates in the limit, there must exist a time $t_j > t_{j-1} + 1$
    where the algorithm outputs an integer $z_{t_j} > z_{t_{j-1}} + (t_j - t_{j-1})$.
    For $t_{j-1} + 2 \le i \le t_j + 1$, we set
    \[x_i =
    \begin{cases}
        z_{t_{j-1}} + (i-t_{j-1}) & t_{j-1} + 2 \le i \le t_j \\
        -(j + 1) & i = t_j + 1
    \end{cases}.\]
    In other words, the adversary follows the enumeration $\seqdef{x}{j}$ from time
    $t_{j-1}+2$ to time $t_j$ and then generates $-(j+1)$ in the next timestep.

    \Cref{fig:union} is an example of what the enumeration $x$ and outputs $z$ may look like at the end of stage $2$.

    To conclude, consider the language $K = \bigcup_{i \in \NN} \{x_i\}$, where the sequence $x$
    is constructed from the infinite iterative procedure described above.
    By construction, $\ZZ_{< 0} \subseteq K$, so $K \in \SC_2$.
    Furthermore, there exists an infinite sequence of times $t_0 < t_1 < \dots$
    where the algorithm's output $z_{t_j}$ at each such time is not contained in $K$.
    Thus $G$ does not generate $K$ in the limit.
\end{proof}

\begin{figure}
    \centering
    \begin{tikzpicture}[x=.93cm,y=0.6cm,>=Latex,
        every label/.style={font=\scriptsize,inner sep=2pt}]

        \draw[<->] (-3.5,0) -- (13.5,0) node[right] {$\mathbb{Z}$};
        \foreach \n in {-3,-2,...,13}
        \draw (\n,0) -- (\n,0.04) node[above=0pt,scale=0.65]{\(\n\)};

        \def\ay{0.75}
        \foreach \x/\idx/\lab in {
            0/0/{$x_0$},
            1/1/{$x_1$},
            -1/2/{$x_2$},
            6/3/{$x_3$},
            7/4/{$x_4$},
            8/5/{$x_5$},
            -2/6/{$x_6$},
            11/7/{$x_7$},
            12/8/{$x_8$},
            -3/9/{$x_9$}}
            {
            \node[fill=blue,circle,inner sep=1.25pt,label=above:\lab] at (\x,\ay) {};
        }


        \def\by{-0.6}
        \foreach \x/\idx/\lab in {
            -1/0/{$z_0$},
            4/1/{$z_1$},
            -2/2/{$z_2$},
            -3/3/{$z_3$},
            2/4/{$z_4$},
            9/5/{$z_5$},
            3/6/{$z_6$},
            10/7/{$z_7$},
            13/8/{$z_8$},
            5/9/{$z_9$}}
            {
            \ifnum
                \idx=1
                \node[fill=red,circle,inner sep=1.5pt,label=below:\lab] at (\x,\by) {};
            \else
                \ifnum
                    \idx=5
                    \node[fill=red,circle,inner sep=1.5pt,label=below:\lab] at (\x,\by) {};
                \else
                    \ifnum
                        \idx=8
                        \node[fill=red,circle,inner sep=1.5pt,label=below:\lab] at (\x,\by) {};
                    \else
                        \node[fill=green,circle,inner sep=1.5pt,label=below:\lab] at (\x,\by) {};
                    \fi
                \fi
            \fi
        }



    \end{tikzpicture}
    \caption{An example of possible values of $x$ and $z$ at the end of stage $2$
        where $t_0 = 1$, $t_1 = 5$, and $t_2 = 8$.
        The red dots represent the incorrect values output by the algorithm
        at times $z_{t_0}$, $z_{t_1}$, and $z_{t_2}$.}
    \label{fig:union}
\end{figure}

    \section{Lossy Generation.}\label{sec:lossy}In this section, we define models of generation where the adversary is allowed to omit some strings in its enumeration
of the target language $K$.

\subsection{Generation Without Samples.}

In the most extreme model, the adversary is allowed to omit the entire target language.
Equivalently, we can think of this model as requiring the algorithm to generate by itself
without receiving any information about the target language.

\begin{definition}[Generator algorithm without samples]
    A generator algorithm without samples is an \emph{injection} $G \colon \NN \to U$,
    where $G(t)$ represents the algorithm's output at time $t$ for every $t \in \NN$.
\end{definition}

Note that in the above definition, we require the function to be an injection.
That is, every output $z_t$ must be unique and cannot appear more than once in the algorithm's output.
This condition ensures that the algorithm cannot repeatedly output a single string,
and is similar to the restriction in generation with samples that the algorithm must
output a string that has not yet appeared in the enumeration.

As before, we get different refinements depending on how we quantify
the timestep at which the algorithm must generate correctly.

\begin{definition}[Uniform generation without samples]
    An algorithm uniformly generates without samples for a collection $\SC$
    if there exists a $\tst$ such that for any language $K \in \SC$ and all $t \ge \tst$,
    the string $z_t$ generated by the algorithm at time $t$ belongs to $K$.
    In addition, there cannot exist distinct times $t \neq t'$ such that $z_t = z_{t'}$.
\end{definition}

Note that in this setting, there is no enumeration of the target language,
so the definitions for non-uniform generation and generation in the limit coincide.

\begin{definition}[Generation in the limit without samples %
(Non‑uniform generation without samples)]
    An algorithm generates in the limit without samples for a collection $\SC$
    if for any language $K \in \SC$, there exists a $\tst$ such that for all $t \ge \tst$,
    the string $z_t$ generated by the algorithm at time $t$ belongs to $K$.
    In addition, there cannot exist distinct times $t \neq t'$ such that $z_t = z_{t'}$.
\end{definition}

We now show that these models are equivalent to the uniform/non-uniform noisy generation settings
(where the time is independent of the noise level) introduced in~\cite{RR25}.
This provides a simpler way of thinking about uniform and non-uniform noisy generation
using a model that does not involve an adversary.
To show this equivalence, we now prove each part of \Cref{thm:equi_lossy_noisy} separately.

\begin{algorithm2e}
    \SetKwInput{Input}{Input}
    \Input{An infinite sequence $z_0$, $z_1$, \dots}

    $i = 0$ \\
    \For{$t = 0, 1, 2, \dots$}{
        \label{line:noisy_for}
        Adversary reveals $x_t$ \\
        $i = \min \{j \ge i \mid z_j \not \in \{x_0, \dots, x_t\}\}$ \\
        \label{line:noisy_min}
        output $z_i$ \\
        $i = i + 1$
    }
    \caption{Noisy generator}
    \label{alg:sim_noisy}
\end{algorithm2e}


\begin{algorithm2e}
    \SetKwInput{Input}{Input}
    \Input{An infinite sequence $z_0$, $z_1$, \dots}

    $i = 0$ \\
    $S = \emptyset$ \\
    \For{$t = 0, 1, 2, \dots$}{
        \label{line:without_for}
        $i = \min \{j \ge i \mid z_j \notin S\}$\\
        \label{line:without_min}
        output $z_i$ \\
        $S = S \cup \{z_i\}$\\
        $i = i + 1$ \\
    }
    \caption{Generator without samples}
    \label{alg:sim_without}
\end{algorithm2e}

\begin{theorem}
    \label{thm:equi_unif_lossy_noisy}
    A collection $\SC$ is uniformly noisily generatable if and only if
    $\SC$ is uniformly generatable without samples.
\end{theorem}

\begin{proof}
    Let $G$ uniformly generate without samples for $\SC$
    and let $z_0$, $z_1$, \dots\ be the strings generated by $G$.
    We claim that \cref{alg:sim_noisy} on input $z_0$, $z_1$, \dots\ is
    a uniform noisy generator for $\SC$.
    Intuitively, \cref{alg:sim_noisy} simply outputs the strings $z_0$, $z_1$, \dots\ , while
    skipping over the strings already enumerated by the adversary.
    First note that the strings $z_i$ are distinct, so the set on line~\ref{line:noisy_min} is always nonempty.
    Since $G$ uniformly generates without samples, there is a time $\tst$ such that $z_t \in K$ for any $t \ge \tst$ and
    $K \in \SC$.
    Because $i$ is incremented each iteration, we have $i \ge t$ at the beginning
    of each iteration on line~\ref{line:noisy_for}.
    Thus if $t \ge \tst$, we have that $i \ge \tst$, implying that $z_i \in K$.
    Furthermore, line~\ref{line:noisy_min} ensures that $z_i \not \in \{x_0, \dots, x_t\}$.
    Thus for any $t \ge \tst$ and $K \in \SC$, the output $z_i$ is in $K \setminus \{x_0, \dots, x_t\}$ as desired.

    For the other direction, fix a uniformly noisily generating algorithm $G$,
    and assume without loss of generality that $\SC$ is a collection over $\NN$.
    Now for each $t$, define $z_t$ to be $G(0, \dots, t)$.
    We claim that \cref{alg:sim_without} on input $z_0$, $z_1$, \dots\ is
    a uniform generator for $\SC$ without samples.
    Intuitively, \cref{alg:sim_without} simply outputs the strings $z_0$, $z_1$, \dots\ , while
    skipping over any repeated strings.
    We first show that the set $\{j \ge i \mid z_j \notin S\}$ on line~\ref{line:without_min} is always nonempty.
    Note that the sequence $0$, \dots, $t$ is a valid beginning to a noisy enumeration of any $K \in \SC$.
    Thus there exists a time step $\tst$ at which $z_t \in K \setminus \{0, \dots, t\}$
    for any $K \in \SC$ and $t \ge \tst$.
    At any iteration of line~\ref{line:without_for}, let $m = \max(S)$.
    For all $t \ge \max(\tst, m)$, note that $z_t \in K \setminus \{0, \dots, t\}$
    and $S \subseteq \{0, \dots, t\}$.
    Thus $z_t \not \in S$, implying that line~\ref{line:without_min} is well-defined.
    Furthermore, we always have that $i \ge t$ during the execution of the algorithm,
    so when $t \ge \tst$, we have $z_i \in K$ for any $K \in \SC$.
    Finally, since $z_i \not \in S$ at each iteration, each of the outputs is unique as desired.
\end{proof}

A very similar proof shows the corresponding theorem for generation in the limit without samples.

\begin{theorem}
    \label{thm:equi_limit_lossy_noisy}
    A collection $\SC$ is non-uniformly noisily generatable if and only if
    $\SC$ is generatable in the limit without samples.
\end{theorem}

\begin{proof}
    Let $G$ generate in the limit without samples for $\SC$
    and let $z_0$, $z_1$, \dots\ be the strings generated by $G$.
    As in the previous proof, we claim that \cref{alg:sim_noisy} on input $z_0$, $z_1$, \dots\
    is a non-uniform noisy generator for $\SC$.
    First note that the strings $z_i$ are distinct, so the set on line~\ref{line:noisy_min} is always nonempty.
    Now fix an arbitrary $K$ and let $\tst$ be the time at which $z_t \in K$ for any $t \ge \tst$.
    Since $i$ is incremented each iteration, we have $i \ge t$ at the beginning
    of each iteration on line~\ref{line:noisy_for}.
    Thus if $t \ge \tst$, we have that $i \ge \tst$, implying that $z_i \in K$.
    Furthermore, line~\ref{line:noisy_min} ensures that $z_i \not \in \{x_0, \dots, x_t\}$.
    Thus for any $t \ge \tst$ and $K \in \SC$, the output $z_i$ is in $K \setminus \{x_0, \dots, x_t\}$ as desired.

    For the other direction, fix a non-uniformly noisily generating algorithm $G$,
    and assume without loss of generality that $\SC$ is a collection over $\NN$.
    Now for each $t$, define $z_t$ to be $G(0, \dots, t)$.
    We claim that \cref{alg:sim_without} on input $z_0$, $z_1$, \dots\
    generates in the limit for $\SC$ without samples.
    We first show that the set $\{j \ge i \mid z_j \notin S\}$ on line~\ref{line:without_min} is always nonempty.
    Fix an arbitrary $K \in \SC$.
    Note that the sequence $0$, \dots, $t$ is a valid beginning to a noisy enumeration of $K$.
    Thus there exists a time step $\tst$ at which $z_t \in K \setminus \{0, \dots, t\}$ for all $t \ge \tst$.
    At any iteration of line~\ref{line:without_for}, let $m = \max(S)$.
    For all $t \ge \max(\tst, m)$, note that $z_t \in K \setminus \{0, \dots, t\}$
    and $S \subseteq \{0, \dots, t\}$.
    Thus $z_t \not \in S$, implying that line~\ref{line:without_min} is well-defined.
    Furthermore, we always have that $i \ge t$ during the execution of the algorithm,
    so when $t \ge \tst$, we have $z_i \in K$.
    Finally, since $z_i \not \in S$ at each iteration, each of the outputs is unique as desired.
\end{proof}

Combining Theorem~$3.1$ in~\cite{RR25} with \cref{thm:equi_unif_lossy_noisy},
we have the following characterization of uniform generation in the limit without samples.

\begin{theorem}
    \label{thm:gen_unif_without}
    A collection $\SC$ is uniformly generatable without samples
    if and only if $\abs{\bigcap_{L \in \SC} L} = \infty$.
\end{theorem}

A similar characterization of non-uniform noisy generation does not appear in the literature.
We bridge this gap by giving a full characterization of collections that are generatable in the limit without samples.
By \cref{thm:equi_limit_lossy_noisy}, this also provides a characterization of non-uniform noisy generation.

\begin{algorithm2e}
    \SetKwInput{Input}{Input}
    \Input{An infinite chain of collections $\SC_0 \subseteq \SC_1 \subseteq \dots$}

    $i = 0$ \\
    $S = \emptyset$ \\
    \For{$t = 0, 1, 2, \dots$}{
        \label{line:without_lim_for}
        $i = \min\{j \ge i \mid j \in \left(\bigcap_{L \in \SC_t} L\right) \setminus S\}$ \\
        \label{line:without_lim_min}
        output $z_t = i$ \\
        $S = S \cup \{i\}$\\
        $i = i + 1$ \\
    }
    \caption{Generator in the limit without samples}
    \label{alg:without_lim}
\end{algorithm2e}

\charnonuniform*


\begin{proof}
    For one direction,
    let $\SC$ be a collection and assume that there exists a countable sequence of collections
    $\SC_0 \subseteq \SC_1 \subseteq \dots$ such that $\SC = \bigcup_{i \in \NN} \SC_i$
    and $\abs{\bigcap_{L \in \SC_i} L} = \infty$ for all $i \in \NN$.
    Assume without loss of generality that $\SC$ is a collection over $\NN$.
    We claim that \Cref{alg:without_lim} on input $\SC_0$, $\SC_1$, \dots\
    generates in the limit without samples.
    During each iteration of the for loop on line~\ref{line:without_lim_for},
    we have by the hypothesis that $\abs{\bigcap_{L \in \SC_t} L} = \infty$.
    Since $S$ is finite, the set $\left(\bigcap_{L \in \SC_t} L\right) \setminus S$ must be infinite.
    Furthermore, since all but a finite number of elements of $\NN$ are greater than $i$,
    the set on line~\ref{line:without_lim_min} must be nonempty.

    Now consider an arbitrary $K \in \SC$.
    Since $\SC = \bigcup_{i \in \NN} \SC_i$, there must be some index $\tst$ such that $K \in \SC_{\tst}$.
    The collections form a chain, so $K$ must be in $\SC_t$ for all $t \ge \tst$.
    Thus for each iteration $t \ge \tst$, the output $z_t$ of \Cref{alg:without_lim}
    is in $\bigcap_{L \in \SC_t} L \subseteq K$ as desired.
    Furthermore, the output $z_t$ is not in $S$, so each output is distinct.
    (Note that if we were to simply choose an arbitrary element from $\SC_t$ at each iteration,
    the condition that each output must be distinct may be violated.)

    For the other direction, fix an algorithm $G$ which generates in the limit without samples.
    For each language $L \in \SC$, define $\tst(\SC, L)$ to be the time at which $G$
    generates correctly for the target language $L$.
    We now construct $\SC_i = \{L \in \SC \mid \tst(\SC, L) \le i \}$ for each $i \in \NN$.
    Clearly $\SC_0 \subseteq \SC_1 \subseteq \dots$.
    Furthermore, since each collection $\SC_i$ is uniformly generatable (in particular at time $i$),
    we have by \Cref{thm:gen_unif_without} that $\abs{\bigcap_{L \in \SC_i} L} = \infty$ for each $i$ as desired.
\end{proof}

For countable collections $\SC$, we provide a simpler version of \cref{thm:gen_limit_without}.

\begin{theorem}
    If $\SC$ is a countable collection, then $\SC$ is generatable in the limit without samples if and only if
    $\abs{\bigcap_{L \in \SC'} L} = \infty$ for every finite $\SC' \subseteq \SC$.
\end{theorem}
\begin{proof}
    First let $\SC$ be a countable collection and assume that
    $\abs{\bigcap_{L \in \SC'} L} = \infty$ for every finite $\SC' \subseteq \SC$.
    If $\SC$ has a finite number of languages, then our hypothesis implies that
    $\abs{\bigcap_{L \in \SC} L} = \infty$, so in fact $\SC$ is uniformly generatable without samples.
    Otherwise, fix an arbitrary ordering of the languages in $\SC$ so that $\SC = \{L_0, L_1, \dots\}$.
    Now for each $i \in \NN$, let $\SC_i = \{L_0, \dots, L_i\}$.
    The collections form a chain $\SC_0 \subseteq \SC_1 \subseteq \dots$
    such that $\SC = \bigcup_{i \in \NN} \SC_i$ and $\abs{\bigcap_{L \in \SC_i} L} = \infty$ for all $i \in \NN$.
    Thus by \Cref{thm:gen_limit_without}, $\SC$ is generatable in the limit without samples.

    For the other direction, let $G$ be an algorithm that generates in the limit without samples,
    and let $\SC' \subseteq \SC$ be an arbitrary \emph{finite} subset of $\SC$.
    Since $G$ generates in the limit without samples, there is a time $t(L)$
    for each $L \in \SC'$ at which time $G$ must generate correctly for $L$.
    Thus at time $\tst = \max_{L \in \SC'} t(L)$, the algorithm $G$ must
    generate correctly for every language in $\SC'$.
    Since each $t(L)$ is finite and $\SC'$ is finite, the time $\tst$ must also be finite.
    Thus $G$ generates uniformly in the limit for $\SC'$.
    By \Cref{thm:gen_unif_without}, we have $\abs{\bigcap_{L \in \SC'} L} = \infty$ as desired.
\end{proof}
%

\subsection{Generation with Infinite Omissions.}
Another natural setting in lossy generation is to require the adversary to enumerate any infinite
subset of the target language.
In particular, this allows the adversary to omit an infinite number of strings
in the target language from the enumeration.

\begin{definition}[Enumeration with infinite omissions]
    For any infinite language $K$, an enumeration of $K$ with infinite omissions
    is any infinite sequence $x_0$, $x_1$, \dots\ without repetitions,
    such that $\bigcup_{i \in \NN} \{x_i\} \subseteq K$.
\end{definition}

We now define the notions of uniform and non-uniform generation with infinite omissions.

\begin{definition}[Uniform generation with infinite omissions]
    An algorithm $G$ uniformly generates with infinite omissions for a collection $\SC$
    if there exists a time step $\tst$ such that for any $K \in \SC$,
    and every enumeration $x$ of $K$ with infinite omission and every time $t \ge \tst$,
    the generated string $z_t$ at time $t$ is in $K \setminus S_t$.
\end{definition}

\begin{definition}[Non-uniform generation with infinite omissions]
    An algorithm $G$ non-uniformly generates with infinite omissions for a collection $\SC$
    if for any $K \in \SC$, there exists a time step $\tst$
    such that for every enumeration $x$ of $K$ with infinite omission and every time $t \ge \tst$,
    the generated string $z_t$ at time $t$ is in $K \setminus S_t$.
\end{definition}

We now provide two conceptual results not presented in the introduction.
We show that for both uniform and non-uniform generation,
allowing the adversary to omit an infinite number of strings does not
change which collections can be generated.
In fact, we show the stronger statement that \emph{any} algorithm which (non)-uniformly
generates for a collection will also (non)-uniformly generate with infinite omissions.

\begin{theorem}
    \label{thm:unif_inf}
    If an algorithm $G$ uniformly generates for a collection $\SC$,
    then $G$ also uniformly generates for $\SC$ with infinite omissions.
\end{theorem}
\begin{proof}
    Let $\SC$ be any collection and let $G$ uniformly generate for $\SC$.
    There must exist a time $\tst$ such that for any $K \in \SC$ and enumeration $x$ of $K$,
    the output $z_t$ is an unseen string of $K$ for all $t \ge \tst$.
    Now consider an arbitrary language $K \in \SC$ and an arbitrary enumeration $x'$ of $K$
    with infinite omissions.
    Note that for any $t$, the prefix $x'_0$, \dots, $x'_{t}$
    is the beginning of some valid full enumeration of $K$.
    Thus for all $t \ge \tst$, the algorithm $G$ must output a valid unseen string of $K$
    when given enumeration $x'$.
\end{proof}

\begin{theorem}
    \label{thm:nonuni_inf}
    If an algorithm $G$ non-uniformly generates for a collection $\SC$,
    then $G$ also non-uniformly generates for $\SC$ with infinite omissions.
\end{theorem}
\begin{proof}
    Let $\SC$ be any collection and let $G$ non-uniformly generate for $\SC$.
    Fix an arbitrary $K \in \SC$.
    There must exist a time $\tst$ such that for any enumeration $x$ of $K$,
    the output $z_t$ is an unseen string of $K$ for all $t \ge \tst$.
    Now consider an arbitrary enumeration $x'$ of $K$ with infinite omissions.
    Note that for any $t$, the prefix $x'_0$, \dots, $x'_{t}$
    is the beginning of some valid full enumeration of $K$.
    Thus for all $t \ge \tst$, the algorithm $G$ must output a valid unseen string of $K$
    when given enumeration $x'$.
\end{proof}

The above two results highlight an important difference between uniform/non-uniform generation
and generation in the limit (in the lossless setting).
Uniform and non-uniform generation require the algorithm to generate correctly after a finite time step,
independent of the enumeration, implying that the algorithm should be able to generate
correctly after only being shown a subset of the language.
In contrast, generation in the limit allows the algorithm to eventually take into account the entire target language.



\subsection{Generation with Finite Omissions.}
In the final setting in lossy generation, the adversary is only allowed to omit a finite number of strings. Since all
languages are infinite, the natural intuition is that finite omissions should not change whether a collection is
generatable.
Quite surprisingly, we show that this intuition is entirely wrong---there are generatable collections in the standard
(lossless) setting that become ungeneratable even if the adversary omits a single string!

We start with some definitions that quantify the number of strings omitted by the adversary.

\begin{definition}[Enumeration with $i$ omissions]
    For any infinite language $K$ and integer $i \in \NN$, an enumeration
    of $K$ with $i$ omissions is any infinite sequence $x_0$, $x_1$, \dots\ without repetition,
    such that $\bigcup_{i \in \NN} \{x_i\} \subseteq K$ and
    $\abs{K \setminus \bigcup_{i \in \NN} \{x_i\}} \le i$.
\end{definition}

\begin{definition}[Generation in the limit with $i$ omissions]
    For any $i \in \NN$, an algorithm $G$ generates in the limit with $i$ omissions
    if for every $K \in \SC$ and every enumeration $x$ of $K$ with $i$ omissions,
    there exists a time $\tst$ such that for all $t \ge \tst$,
    the string generated by the algorithm at time $t$ belongs to $K \setminus \{x_0, \dots x_t\}$.
\end{definition}

We could also define versions of uniform and non-uniform generation with finite omissions.
However, \Cref{thm:unif_inf,thm:nonuni_inf} already show that
(non)-uniform generation with infinite omissions is equivalent to (non)-uniform generation without loss.
Thus, we only consider generation in the limit in the setting with finite omissions.
We now state and prove the portion of \Cref{thm:fine_grained} concerning lossy generation.

\begin{restatable}{theorem}{omitone}
    \label{thm:omit_one}
    For every $i \in \NN$, there exists a collection that is generatable in the limit
    with $i$ omissions, but is not generatable in the limit with $i+1$ omissions.
    Therefore, for $i=0$ in particular,
    there is a collection that is generatable in the limit, but not so if a single string is omitted.
\end{restatable}


We prove \Cref{thm:omit_one} in two parts.
For any $i \in \NN$, define
$\SC_1^i = \bigcup_{j \in \NN} \{\{0, \dots i\} \cup A \cup P_j \mid A \subseteq \ZZ\}$,
$\SC_2^i = \{A \cup \ZZ_{< 0} \mid A \subseteq \ZZ \setminus \{0, \dots i\}\}$,
and $\SC^i = \SC_1^i \cup \SC_2^i$.
We first show that $\SC^i$ can be generated with $i$ omissions.

\begin{lemma}
    \label{lem:omit_i}
    For any $i \in \NN$, the collection $\SC^i$ is generatable in the limit with $i$ omissions.
\end{lemma}
\begin{proof}
    Consider the algorithm $G$ whose outputs are given by
    \[
        G(x_0, \dots, x_t) =
        \begin{cases}
            \max\left(\{t, x_0, z_0, \dots, x_t\}\right) + 1 & (\{0, \dots, i\} \cap S_t) \neq \emptyset \\
            \min(\{0, x_0, z_0, \dots, x_t\}) -1 & \text{otherwise}
        \end{cases}.\]
    We now consider the two possible cases where either $K \in \SC_1^i$ or $K \in \SC_2^i$.

    In the case where $K \in \SC_1^i$, there must be some $j$ such that $P_j \subseteq K$.
    In addition, since $\{0, \dots, i\} \subseteq K$, and the algorithm can omit at most $i$ strings,
    there must be some time $t'$ in every enumeration when $(\{0, \dots, i\} \cap S_{t'}) \neq \emptyset$.
    We claim that after time $\tst = \max(j, t')$, the algorithm $G$ always generates correct strings.
    For any time $t \ge \tst$, we have $(\{0, \dots, i\} \cap S_t) \neq \emptyset$, so $G$
    outputs $z_t = \max\left(\{t, x_0, z_0, \dots, x_t\}\right) + 1$.
    In particular $z_t$ is an unseen integer that is at least $t > j$.
    Since $P_j \subseteq K$, every integer larger than $j$ is in $K$, implying that
    $z_t$ is an unseen integer contained in $K$.

    In the other case where $K \in \SC_2^i$, note that $\{0, \dots, i\} \cap K = \emptyset$,
    so the algorithm's output at time $t$ is $z_t = \min(\{0, x_0, z_0, \dots, x_t\}) - 1$.
    In particular, $z_t$ is always a negative unseen string.
    Since $\ZZ_{< 0} \subseteq K$, the output $z_t$ must be a correct unseen string.
    Thus $G$ generates in the limit with $i$ omissions.
\end{proof}

We now show that $\SC^i$ cannot be generated with $i+1$ omissions.
This proof is very similar to the proof of \Cref{thm:union_ungen}.

\begin{lemma}
    \label{lem:omit_i_1}
    For any $i \in \NN$, the collection $\SC^i$ cannot be generated in the limit with $i+1$ omissions.
\end{lemma}
\begin{proof}
    Assume for contradiction that there exists an algorithm $G$ which generates in the limit for $\SC$.
    We will inductively construct an enumeration $\{x_k\}_{k \in \NN}$ of a language $K \in \SC_2$ such that there
    exists an infinite sequence of times $t_0 < t_1 < \dots$ where the string output by $G$ at each time $t_i$
    is not in $K$.
    This would imply that $G$ does not generate in the limit for $K$.

    Let $L_0 = \NN$ and consider the adversary enumeration
    $\seqdef{x}{0}$ where $\seq{x}{0}{k} = k + i + 1$.
    Since $\seqdef{x}{0}$ forms an enumeration of $\NN$ with $i+1$ omissions, there must exist a time $t_0$
    where the algorithm outputs a string $z_{t_0} \in \NN \setminus \{\seq{x}{0}{0}, \dots, \seq{x}{0}{t_0}\}$.
    For $k \le t_0 + 1$, we set
    \[x_k =
    \begin{cases}
        k + i + 1 & k \le t_0 \\
        -1 & k = 1 + t_0
    \end{cases}.\]

    We now proceed iteratively in stages, where during each stage $j$,
    we extend the (partial) enumeration $x$ and introduce a new time $t_j$ for which the algorithm's
    output $z_{t_j}$ is incorrect.
    At the beginning of stage $j + 1$, let $t_0 < \dots < t_j$ be the current increasing sequence of times
    where the algorithm makes a mistake,
    and let $x_0$, \dots, $x_{1 + t_j}$ be the current partial enumeration.
    Recall that we write $S_t$ to denote the enumerated strings $\{x_0, \dots, x_t\}$.
    Inductively assume that $\{0, -1, \dots, -(j+1)\} \subseteq S_{1 + t_j}$,
    and that for each of the times $t_i$, the algorithm's output $z_{t_i}$ is a \emph{nonnegative}
    integer which is not contained in $S_{1 + t_j}$.

    Let $m_t = \max\{x_0, z_0, \dots, x_t, z_t\}$ be the maximum integer
    output by either the adversary or the algorithm up to time $t$.
    Now consider the language
    $L_{j+1} = \{0, \dots, i\} \cup S_{1 + t_j} \cup \{1 + m_{1 + t_j}, 2 + m_{1 + t_j}, \dots\}$.
    Clearly $L_{j+1} \in \SC_1$.
    Thus consider the enumeration $\seqdef{x}{j+1}$ where
    \[\seq{x}{j+1}{i} = \begin{cases}
                            x_i & i \le 1 + t_j \\
                            m_{1 + t_j} + i - 1 - t_j & i \ge 2 + t_j
    \end{cases}.\]
    Since $\seqdef{x}{j+1}$ is an enumeration of $L_{j+1} \setminus \{0, \dots, i\}$,
    the enumeration $\seqdef{x}{j+1}$ is an enumeration of $L_{j+1}$ with $i+1$ omissions.
    Since $G$ generates in the limit with $i+1$ omissions, there must exist a time $t_{j+1} > 1 + t_j$
    where the output $z_{t_{j+1}}$ is in
    $L_{j+1} \setminus \{\seq{x}{j+1}{0}, \dots, \seq{x}{j+1}{t_{j+1}}\}$.
    Thus we extend the enumeration $x$ for $i \in [2 + t_j, 1 + t_{j+1}]$ by setting
    \[x_i =
    \begin{cases}
        \seq{x}{j+1}{i} & 2 + t_j \le i \le t_{j+1} \\
        -(j+2) & i = 1 + t_{j+1}
    \end{cases}.\]
    It remains to check that the inductive hypotheses are satisfied.
    First, since we do not change the sequence $x_0$, \dots, $x_{1 + t_{j}}$,
    the algorithm's output up until time $1 + t_j$ remains the same.
    In particular, the values of $z_{t_0}$, \dots, $z_{t_j}$ remain unchanged.
    Furthermore, since the additional enumerated nonnegative strings $x_{2 + t_j}$, \dots, $x_{t_{j+1}}$
    are all greater than $m_{1 + t_j}$, it remains true that each of the outputs $z_{t_i}$
    are not contained in $S_{1 + t_{j+1}}$.
    It is also clear that $\{0, -1, \dots, -(j + 2)\} \subseteq S_{1 + t_{j+1}}$.
    Finally, the new algorithm output $z_{t_{j + 1}}$ is both nonnegative and not contained in
    $S_{1 + t_{j+1}}$.

    To conclude, consider the language $K = \bigcup_{k \in \NN} \{x_k\}$, where the sequence $x$
    is constructed from the infinite iterative procedure described above.
    By construction, $\ZZ_{< 0} \subseteq K$ and $\{0, \dots, i\} \cap K = \emptyset$, so $K \in \SC_2$.
    Furthermore, there exists an infinite sequence of times $t_0 < t_1 < \dots$
    where the algorithm's output $z_{t_k}$ at each such time is not contained in $K$.
    Thus $G$ does not generate in the limit.
\end{proof}

We now combine the above lemmas to establish \Cref{thm:omit_one}.

\begin{proof}[Proof of \Cref{thm:omit_one}]
    For any $i \in \NN$, by \Cref{lem:omit_i} and \Cref{lem:omit_i_1}, the collection $\SC^i$
    is generatable in the limit with $i$ omissions, but is not generatable in the limit with $i+1$ omissions.
\end{proof}

    \section{Generation with Noise.}\label{sec:noisy}Recall that in the noisy model defined in~\cite{RR25},
the adversary is allowed to pick a noise level $\nst$
and insert $\nst$ extraneous strings into its enumeration.
Crucially, the algorithm is not told about the number of strings inserted into the enumeration.
In this section, we explore a more fine-grained notion of generation with noise
for each noise-level where the algorithm is informed about the number of inserted strings.

%
%
%


\begin{definition}[Noisy enumeration with noise level $i$]
    For any infinite language $K$ and integer $i \in \NN$, a noisy enumeration
    of $K$ with noise level $i$ is any infinite sequence $x_0$, $x_1$, \dots\ without repetitions,
    such that $K \subseteq \bigcup_{i \in \NN} \{x_i\}$ and
    $\abs{\bigcup_{i \in \NN} \{x_i\} \setminus K} \le i$.
\end{definition}

We now define the notions of generation in the limit and non-uniform generation for each noise level.

\begin{definition}[Generation in the limit with noise level $i$]
    For any $i \in \NN$, an algorithm $G$ generates in the limit with noise level $i$
    if for every $K \in \SC$ and every enumeration $x$ of $K$ with noise level at most $i$,
    there exists a time $\tst$ such that for all $t \ge \tst$,
    the string generated by the algorithm at time $t$ belongs to $K \setminus \{x_0, \dots x_t\}$.
\end{definition}

\begin{definition}[Non-uniform generation with noise level $i$]
    For any $i \in \NN$, an algorithm $G$ non-uniformly generates with noise level $i$
    if for every $K \in \SC$, there exists a time $\tst$ such that for
    every enumeration $x$ of $K$ with noise level at most $i$ and all $t \ge \tst$,
    the string generated by the algorithm at time $t$ belongs to $K \setminus \{x_0, \dots x_t\}$.
\end{definition}



Clearly, a collection which is generatable in the limit with noise level $1$ is generatable in the limit without noise,
and a collection which is noisily generatable in the limit is
generatable with noise level $i$ for every $i$.
In fact, we show that both of these containments are strict.
We first prove the following lemma which gives a necessary condition
for generating in the limit with noise level $i$.

\begin{lemma}
    \label{lem:noise_i_condition}
    Let $\SC$ be a collection over a countable universe $U$.
    If $\SC$ is generatable in the limit with noise level $i$,
    then for any universe $U'$ with $\abs{U \setminus U'} \le i$,
    the projection of $\SC$ onto $U'$ is generatable in the limit.
\end{lemma}

\begin{proof}
    For any integer $i$ and collection $\SC$ over a universe $U$,
    let $G$ be an algorithm that generates in the limit with noise level $i$ for $\SC$.
    Let $U'$ be any universe with $\abs{U \setminus U'} \le i$ and let $\SC'$
    be the projection of $\SC$ onto $U'$.
    If $\SC'$ is empty, then the claim is trivially satisfied.
    Otherwise, let $d = \abs{U \setminus U'}$ and arbitrarily index the elements in $U \setminus U'$
    by $y_1$, \dots, $y_d$.
    Consider the algorithm $G'$ whose outputs are given by
    \[G'(x_{0:t}) = G(y_1, \dots, y_d, x_0, \dots, x_t).\]
    We claim that $G'$ generates in the limit for $\SC'$.

    Let $K'$ be an arbitrary language in $\SC'$ and $x$ be an arbitrary enumeration of $K'$.
    Since $\SC'$ is the projection of $\SC$ onto $U'$, there must exist a language $K \in \SC$
    such that $K' = K \cap U'$.
    Now let $E = \{y_1, \dots, y_d, x_0, x_1, \dots\}$.
    Since $x$ forms an enumeration of $K'$ and $\{y_1, \dots, y_d\} = U \setminus U'$,
    we have $E = K' \cup (U \setminus U')$.
    Furthermore, since $K' \subseteq K$ and $\abs{U \setminus U'} \le i$,
    we have that $\abs{E \setminus K} \le i$.
    Thus the sequence $y_1$, \dots, $y_d$, $x_0$, $x_1$, \dots\ is a noisy
    enumeration of $K$ with noise level at most $i$.

    Since $G$ generates in the limit with noise level $i$, there must be an index $\tst$
    such that \[G(y_1, \dots, y_d, x_0, \dots, x_t) \in K \setminus \{y_1, \dots, y_d, x_0, \dots, x_t\}\]
    for all $t \ge \tst$.
    Now note that $\{y_1, \dots, y_d, x_0, \dots, x_t\} = (U \setminus U') \cup \{x_0, \dots, x_t\}$.
    Thus we have that $K \setminus \{y_1, \dots, y_d, x_0, \dots, x_t\} = K' \setminus \{x_0, \dots, x_t\}$.
    This implies that for all times $t \ge \tst$, the output $G'(x_{0:t})$ is in $K' \setminus \{x_0, \dots, x_t\}$.
    Thus $G'$ generates in the limit for $\SC'$, completing the proof.
\end{proof}

We also need the following concepts of mappings and isomorphisms between collections.

%

\begin{definition}
    Let $\SC$ be a collection over a universe $U$ and let $f \colon U \to U'$ be any function.
    For any $L \in \SC$, we define $f(L)$ to be the set $\{f(x) \mid x \in L\}$,
    which simply maps every element in $L$ according to $f$.
    Similarly, we define the collection $f(\SC) = \{f(L) \mid L \in \SC\}$.
\end{definition}

\begin{definition}[Isomorphism]
    Let $\SC$ be a collection over a universe $U$ and $\SC'$ be a collection over a universe $U'$.
    We say that $\SC$ is isomorphic to $\SC'$ if there exists a bijection $f \colon U \to U'$
    such that $f(\SC) = \SC'$.
\end{definition}
Clearly, if two collections are isomorphic, then they are both generatable in the limit, or neither are.

We now prove the portion of \Cref{thm:fine_grained} concerning noise,
showing a separation between noise levels $i$ and $i+1$ for every $i$.
This also resolves a question of~\cite{RR25}, in which they ask if there exists a language
that is generatable in the limit without noise, but is not noisily generatable in the limit.
This question is answered in the affirmative by setting $i = 0$.
Interestingly, we note that the collection $\SC^i$ used in the following
proof is identical to the collection $\SC^i$ used in the proof of \Cref{thm:omit_one}.

\begin{restatable}{theorem}{noisehierarchy}
    \label{thm:noise_hierarchy}
    For every integer $i \in \NN$, there exists a collection which is generatable
    in the limit with noise level $i$, but is not generatable in the limit with noise level $i+1$.
    Therefore, for $i=0$ in particular, there is a collection that is generatable in the limit,
    but not so if a single incorrect string is output by the adversary.
\end{restatable}

\begin{proof}
    For any $i \in \NN$, let
    $\SC_1^i = \bigcup_{j \in \NN} \{\{0, \dots i\} \cup A \cup P_j \mid A \subseteq \ZZ\}$,
    $\SC_2^i = \{A \cup \ZZ_{< 0} \mid A \subseteq \ZZ \setminus \{0, \dots i\}\}$,
    and $\SC^i = \SC_1^i \cup \SC_2^i$.
    We first construct an algorithm to show that $\SC^i$ is generatable in the limit with noise level $i$.
    Consider the algorithm $G$ whose outputs are given by
    \[
        G(x_0, \dots, x_t) =
        \begin{cases}
            \max\left(\{t, x_0, z_0, \dots, z_{t-1}, x_t\}\right) + 1 & \{0, \dots, i\} \subseteq S_t \\
            \min(\{0, x_0, z_0, \dots, z_{t-1}, x_t\}) -1 & \text{otherwise}
        \end{cases}.\]
    We now consider the two possible cases where either $K \in \SC_1^i$ or $K \in \SC_2^i$.

    In the case where $K \in \SC_1^i$, there must be some $j$ such that $P_j \subseteq K$.
    In addition, since $\{0, \dots, i\} \subseteq K$, there must be some time $t'$ in every enumeration
    when $\{0, \dots, i\} \subseteq S_{t'}$.
    We claim that after time $\tst = \max(j, t')$, the algorithm $G$ always generates correct strings.
    For any time $t \ge \tst$, we have $\{0, \dots, i\} \subseteq S_t$, so $G$
    outputs $z_t = \max\left(\{t, x_0, z_0, \dots, z_{t-1}, x_t\}\right) + 1$.
    In particular, $z_t$ is an unseen integer which is at least $t > j$.
    Since $P_j \subseteq K$, every integer larger than $j$ is in $K$, implying that
    $z_t$ is an unseen integer contained in $K$.

    In the other case where $K \in \SC_2^i$, note that $\{0, \dots, i\} \cap K = \emptyset$.
    Since the adversary can insert at most $i$ strings into the enumeration,
    there will never exist a time when $\{0, \dots, i\} \subseteq S_t$.
    Thus at any time $t$, the algorithm's output is $z_t = \min(\{0, x_0, z_0, \dots, z_{t-1}, x_t\}) - 1$.
    In particular, $z_t$ is always a negative unseen string.
    Since $\ZZ_{< 0} \subseteq K$, the output $z_t$ must be a correct unseen string.
    Thus $G$ generates in the limit with noise level $i$.

    To see that $\SC^i$ is not generatable in the limit with noise level $i+1$,
    \cref{lem:noise_i_condition} implies that it suffices to show that the projection of $\SC^i$
    onto $\ZZ \setminus \{0, \dots, i\}$ is not generatable in the limit.
    The projection of $\SC^i$ onto $\ZZ \setminus \{0, \dots, i\}$
    is exactly $B^i = B_1^i \cup B_2^i$ where
    $B_1^i = \bigcup_{j = i+1}^{n} \{A \cup P_j \mid A \subseteq \ZZ \setminus \{0, \dots, i\}\}$ and
    $B_2^i = \{A \cup \ZZ_{< 0} \mid A \subseteq \ZZ \setminus \{0, \dots, i\}\}$.
    Now consider the collections $\SC_1 = \bigcup_{i \in \NN} \{A \cup P_i \mid A \subseteq \ZZ\}$,
    $\SC_2 = \{A \cup \ZZ_{< 0} \mid A \subseteq \ZZ\}$, and $\SC = \SC_1 \cup \SC_2$
    used in the proof of~\Cref{thm:union_ungen}.
    Since $\SC$ is not generatable in the limit, it suffices to show that $B^i$ and $\SC$
    are isomorphic.
    Consider the bijection $f \colon \ZZ \to \ZZ \setminus \{0, \dots, i\}$ given by
    \[f(x) =
    \begin{cases}
        x & x < 0 \\
        x + i + 1 & x \ge 0
    \end{cases}.\]
    It is easy to see that $f(\SC) = B^i$, so the two collections are isomorphic.
    Thus since $\SC$ is not generatable in the limit, $B^i$ is also not generatable in the limit as desired.
\end{proof}

We now show that there exists a collection that is generatable in the limit with $i$ elements of noise
for every $i \in \NN$, but is not generatable with noise (in the original model
where the algorithm is not told the noise level).
This implies that knowledge about the noise level does give an algorithm more power.

\noisesensitive*

\eat{
    \begin{theorem}
        There exists a collection which is generatable in the limit with noise level $i$ for any $i \in \NN$,
        but is not noisily generatable in the limit.
    \end{theorem}
}

\begin{proof}
    Consider the collections
    $\SC_1 = \{P_i \mid i \in \NN\}$,
    $\SC_2 = \{A \cup \ZZ_{< 0} \mid A \subseteq \NN\}$,
    and $\SC = \SC_1 \cup \SC_2$.
    We first show that $\SC$ is generatable in the limit with noise level $i$ for any $i \in \NN$.
    For an arbitrary $i \in \NN$, consider the algorithm $G_i$ whose outputs are given by
    \[
        G_i =
        \begin{cases}
            \min(\{0, x_0, z_0, \dots, x_t\}) -1 & \{-1, \dots, -(i+1)\} \subseteq S_t \\
            \max\left(\{t, x_0, z_0, \dots, x_t\}\right) + 1 & \text{otherwise}
        \end{cases}
        .\]
    We consider the two cases where $K \in \SC_1$ or $K \in \SC_2$.

    In the case where $K \in \SC_2$, the adversary must eventually enumerate
    $\ZZ_{< 0}$, so there must exist some time $t'$ such that $\{-1, \dots, -(i+1)\} \subseteq S_t$.
    For all $t > t'$, the algorithm outputs a negative unseen integer.
    Since $\ZZ_{< 0} \subseteq K$, the algorithm $G_i$ generates in the limit.

    In the other case where $K \in \SC_1$, note that $\{-1, \dots, -(i+1)\} \cap K = \emptyset$.
    Since the adversary is allowed to insert at most $i$ strings in its enumeration,
    there will never be a time when $\{-1, \dots, -(i+1)\} \subseteq S_t$.
    Thus, for all times $t$, we have $z_t = \max\left(\{t, x_0, z_0, \dots, x_t\}\right) + 1$.
    Since $K \in \SC_1$, there is some $j$ such that $K = P_j$.
    For all times $t > j$, we can see that $z_t$ is an unseen string that is at least $j$.
    Thus $z_t \in K$ for all times $t > j$, so $G_i$ generates in the limit.

    We now show that $\SC$ is not generatable in the limit with noise.
    This portion of the proof is once again very similar to the proof of \Cref{thm:union_ungen}.
    Assume for contradiction that there exists an algorithm $G$ which noisily generates in the limit for $\SC$.
    We will inductively construct an enumeration $\{x_i\}_{i \in \NN}$ of a language $K \in \SC_2$ such that there
    exists an infinite sequence of times $t_0 < t_1 < \dots$ where the string output by $G$ at each time $t_i$
    is not in $K$.
    This would imply that $G$ does not generate in the limit for $K$.

    Let $L_0 = \NN$ and consider the adversary enumeration
    $\seqdef{x}{0}$ where $\seq{x}{0}{i} = i$.
    Since $G$ generates in the limit, there must exist a time $t_0$
    where the algorithm outputs a string $z_{t_0} \in \NN \setminus \{0, \dots, t_0\}$.
    For $i \le t_0 + 1$, we set
    \[x_i =
    \begin{cases}
        i & i \le t_0 \\
        -1 & i = 1 + t_0
    \end{cases}.\]

    We now proceed iteratively in stages, where during each stage $j$,
    we extend the (partial) enumeration $x$ and introduce a new time $t_j$ for which the algorithm's
    output $z_{t_j}$ is incorrect.
    At iteration $j$, let $t_0 < \dots < t_j$ be the current increasing sequence of times
    and let $x_0$, \dots, $x_{1 + t_j}$ be the current partial enumeration.
    Recall that we write $S_t$ to denote the enumerated strings $\{x_0, \dots, x_t\}$.
    Inductively assume that $\{0, -1, \dots, -(j+1)\} \in S_{1 + t_j}$,
    and that for each of the times $t_i$, the algorithm's output $z_{t_i}$ is a \emph{positive}
    integer which is not contained in $S_{1 + t_j}$.

    Let $m_t = \max\{x_0, z_0, \dots, x_t, z_t\}$ be the maximum integer
    output by either the adversary or the algorithm up to time $t$.
    Now consider the language $L_{j+1} = S_{1 + t_j} \cup \{1 + m_{1 + t_j}, 2 + m_{1 + t_j}, \dots\}$.
    It may not be true that $L_{j+1} \subseteq \SC_1$.
    However, $\{1 + m_{1 + t_j}, 2 + m_{1 + t_j}, \dots\} = L_{j + 1} \setminus S_{1 + t_j}$ is in $\SC_1$.
    Thus consider the natural enumeration $\seqdef{x}{j+1}$ of $L_{j+1}$ where
    \[\seq{x}{j+1}{i} = \begin{cases}
                            x_i & i \le 1 + t_j \\
                            m_{1 + t_j} + i - 1 - t_j & i \ge 2 + t_j
    \end{cases}.\]
    Since $\abs{S_{1 + t_j}} = 2 + t_j$, an enumeration of $L_{j+1}$
    is an enumeration of $\{1 + m_{1 + t_j}, 2 + m_{1 + t_j}, \dots\}$ with noise level $2 + t_j < \infty$.
    Since $G$ noisily generates in the limit, there must exist a time $t_{j+1} > 1 + t_j$
    where the output $z_{t_{j+1}}$ is in
    $\{1 + m_{1 + t_j}, 2 + m_{1 + t_j}, \dots\} \setminus \{\seq{x}{j+1}{0}, \dots, \seq{x}{j+1}{t_{j+1}}\}$.
    Thus we extend the enumeration $x$ for $i \in [2 + t_j, 1 + t_{j+1}]$ by setting
    \[x_i =
    \begin{cases}
        \seq{x}{j+1}{i} & 2 + t_j \le i \le t_{j+1} \\
        -(j+2) & i = 1 + t_{j+1}
    \end{cases}.\]
    It remains to check that the inductive hypotheses are satisfied.
    First, since we do not change the sequence $x_0$, \dots, $x_{1 + t_{j}}$,
    the algorithm's output up until time $1 + t_j$ remains the same.
    In particular, the values of $z_{t_0}$, \dots, $z_{t_j}$ remain unchanged.
    Furthermore, since the additional enumerated positive strings $x_{2 + t_j}$, \dots, $x_{t_{j+1}}$
    are all greater than $m_{1 + t_j}$, it remains true that each of the outputs $z_{t_i}$
    are not contained in $S_{1 + t_{j+1}}$.
    It is also clear that $\{0, -1, \dots, -(j + 2)\} \in S_{1 + t_{j+1}}$.
    Finally, the new algorithm output $z_{t_{j + 1}}$ is both positive and not contained in
    $S_{1 + t_{j+1}}$.

    To conclude, consider the language $K = \bigcup_{i \in \NN} \{x_i\}$, where the sequence $x$
    is constructed from the infinite iterative procedure described above.
    By construction, $\ZZ_{< 0} \subseteq K$, so $K \in \SC_2$.
    Furthermore, there exists an infinite sequence of times $t_0 < t_1 < \dots$
    where the algorithm's output $z_{t_i}$ at each such time is not contained in $K$.
    Thus $G$ does not noisily generate in the limit.
\end{proof}

    \section{Identification and Generation with Feedback.}\label{sec:feedback}In this section, we consider language generation and identification in the presence of feedback.

\subsection{Generation in the Limit with Feedback.}
We largely follow the model of generation in the limit with feedback defined in~\cite{CP25},
which we summarize here.
As in original mode of generation in the limit, the adversary selects a target language $K \in \SC$.
Then at each time step $t$, the adversary outputs a string $x_t \in K$.
The algorithm is now allowed to query a string $y_t$,
and receives a response $a_t \in \{\text{Yes}, \text{No}\}$ corresponding
to whether $y_t$ is in the target language $K$.
Finally, the algorithm outputs a string $z_t$.
As before, we wish for there to be some time $\tst$ after which all
the generated strings $z_t$ are correct, i.e., they are in the target language $K$.

Charikar and Pabbaraju~\cite{CP25} formally define the feedback model
as a game between adversary and generator strategies.
However, we define an equivalent model in terms of enumerations and generator algorithms
which is closer to the standard models of generation.

\begin{definition}[Generator algorithm with feedback]
    A generator algorithm with feedback is a function $G$ that takes as input an alternating sequence
    of strings $x_t$ and responses $a_t$, and generates either a query string $y_t$
    or an output string $z_t$.
    If the input sequence ends with an enumerated string $x_t$,
    then the generator output $G(x_0, a_0, \dots, a_{t-1}, x_t)$
    represents the query string $y_t$.
    Otherwise, the generator output $G(x_0, a_0, \dots, x_t, a_t)$
    represents the output string $z_t$.
\end{definition}

Note that we do not need to include the algorithm's previous queries $y_t$
in the input since the algorithm can reconstruct all of its previous outputs and queries.

\begin{definition}[Generation in the limit with feedback]
    A generator algorithm with feedback $G$ generates in the limit with feedback for a collection $\SC$
    if for any $K \in \SC$ and any enumeration $x$ of $K$,
    there exists a time step $\tst$ where for all $t \ge \tst$,
    the generated string $z_t$ at time $t$ is in $K \setminus S_t$.
\end{definition}

We now prove the first part of \Cref{thm:feedback_combined}, that infinite feedback
is strictly more powerful than generation in the limit.
We do this by showing that the countable union of uniformly generatable collections
is generatable in the limit with feedback.
Since the countable union of uniformly generatable collections is not necessarily
generatable in the limit~\cite{LRT25}(Lemma~$4.3$), there exist
collections which are generatable in the limit with feedback,
but are not generatable in the limit without feedback.

\begin{algorithm2e}[!ht]
    \SetKwInput{Input}{Input}
    \Input{A countable sequence of uniformly generatable collections $\SC_0$, $\SC_1$, \dots}

    $S = \emptyset$ \\
    $t = 0$ \\
    $v = 0$ \\
    \For{$i = 0, 1, 2, \dots$}{\label{line:feed_for}
        $c_i = \text{closure dimension of $\SC_i$}$ \\
        \While{$\abs{S} \le c_i$}{\label{line:feed_while1}
        Adversary reveals $x_t$ \\
            $S = S \cup \{x_t\}$ \\
            Algorithm queries $y_t = v$ \\
            Adversary responds $a_t$ \\
            Output $z_t = v$ \\
            $t = t + 1$ \\
        }
        $A_i = \{L \in \SC_i \mid S \subseteq L\}$ \\
        \If{$A_i \neq \emptyset$}{\label{line:feed_if}
            $B_i = \bigcap_{L \in A_i} L$ \\
            \While{true}{\label{line:feed_while_true}
            Adversary reveals $x_t$ \\
                $S = S \cup \{x_t\}$ \\
                $v = \min\{j \ge v \mid j \in B_i \setminus S\}$ \\
                \label{line:feed_min}
                Algorithm queries $y_t = v$ \label{line:feed_que_v}\\
                Adversary responds $a_t$ \\
                Output $z_t = v$ \\
                \If{$a_t = \text{No}$}{
                    $t = t + 1$ \\
                    \textbf{break}
                }
                $t = t + 1$ \\
            }
        }
    }
    \caption{Generator in the limit with feedback}
    \label{alg:feedback}
\end{algorithm2e}


\begin{theorem}
    \label{thm:feedback_union}
    A collection $\SC$ is generatable in the limit with feedback if there exists a countable set of classes
    $\SC_0$, $\SC_1$, \dots\ such that $\SC = \bigcup_{i \in \NN} \SC_i$ and each $\SC_i$ is uniformly generatable.
\end{theorem}

\begin{proof}
    Let $\SC$ be an arbitrary collection and $\SC_0$, $\SC_1$, \dots\ be a countable sequence of
    collections such that each $\SC_i$ is uniformly generatable and $\SC = \bigcup_{i \in \NN} \SC_i$.
    Assume without loss of generality that $\SC$ is a collection over $\NN$.
    We claim that \Cref{alg:feedback} when run on $\SC_0$, $\SC_1$, \dots\ generates
    uniformly with feedback for $\SC$.
    Intuitively, the algorithm iterates through the set of collections
    and tries to determine if each $\SC_i$ contains the target language $K$.
    For each $\SC_i$, we first make enough feedback queries so that the
    set of positive examples $S$ is at least the closure dimension of $\SC_i$.
    Then, if there are languages consistent with $S$ in $\SC_i$,
    we infinitely output unseen strings from the closure of $S$,
    while making feedback queries from the same set.
    If $K$ is indeed contained in $\SC_i$, we will generate from the closure of $S$ in $\SC_i$ forever.
    Otherwise, the feedback queries will eventually let us know that $K \notin \SC_i$
    and the algorithms moves on to $\SC_{i+1}$.

    Formally, fix an arbitrary $K \in \SC$ and an arbitrary enumeration $x$ of $K$.
    Now consider the beginning of an arbitrary iteration of the for loop on line~\ref{line:feed_for}.
    Since the closure dimension of any uniformly generatable collection is finite~\cite{LRT25},
    the while loop on line~\ref{line:feed_while1} must terminate.
    Now recall that the closure dimension $c_i$ satisfies the property that
    for every $S \subseteq \NN$ with $\abs{S} > c_i$,
    either no languages in $\SC_i$ are consistent with $\abs{S}$,
    or the closure of $S$ in $\SC_i$ is infinite.
    This implies that if we move into the body of the if statement on line~\ref{line:feed_if},
    we must have that $\abs{B_i} = \infty$.

    Now assume that we are in the body of the if statement on line~\ref{line:feed_if}
    and let $k$ be the current value of the variable $v$.
    We claim that if for all elements $n \ge k$, we have that $n \in B_i$ implies $n \in K$,
    then the while loop on line~\ref{line:feed_while_true} will run forever,
    and every output $z_t$ will be in $K \setminus S_t$.
    Firstly, the set on line~\ref{line:feed_min} is nonempty since $\abs{B_i \setminus S} = \infty$,
    and all but a finite number of elements of $\NN$ are larger than $v$.
    Thus on line~\ref{line:feed_que_v}, we have $v \in B_i$.
    Assuming that $n \in B_i$ implies $n \in K$ for all $n \ge k$,
    this would imply that $v \in B_i \setminus S$
    and the output $z_t$ is a correct unseen string as desired.
    In the opposite case where there is some $n \ge k$ such that $n \in B_i$ does not imply $n \in K$,
    the while loop on line~\ref{line:feed_while_true} must eventually break.
    This is because there must be some element $b \ge k$ in $B_i \setminus K$.
    Eventually, the variable $v$ will take on the value $b$, and the query $a_t$
    will return ``No'', causing the while loop to break.

    We have shown that for any iteration of the for loop, if there exists some $n \ge k$
    such that $n \in B_i \setminus K$, then the while loop on line~\ref{line:feed_while_true}
    must break and the for loop will eventually advance to the next iteration.
    In the opposite case, the while loop will iterate forever, and output correct unseen strings
    at each iteration.
    Thus it remains to show that at some iteration $i$ of the for loop,
    it is true that $n \in B_i$ implies $n \in K$ for all $n \ge k$.
    Since $\SC = \bigcup_{i \in \NN} \SC_i$, there must be some index $j$ such that $K \in \SC_j$.
    At such an iteration $j$, we must have $A_i \neq \emptyset$ and $B_i \subseteq K$.
    Thus $n \in B_i$ implies $n \in K$ for all $n \ge k$ as desired,
    and the algorithm must generate in the limit.
\end{proof}

%

Since non-uniformly generatable collections can be written as the countable union of uniformly
generatable collections, we have the following corollary.

\begin{corollary}
    A collection $\SC$ is generatable in the limit with feedback if there exists a countable set of classes
    $\SC_1$, $\SC_2$, \dots\ such that $\SC = \bigcup_{i \in \NN} \SC_i$ and each $\SC_i$ is non-uniformly generatable.
\end{corollary}
\begin{proof}
    \cite{LRT25} (Theorem~$3.5$) showed that every non-uniformly generatable collection can be written
    as the countable union of uniformly generatable collections.
    Since the countable union of countable sets is still countable, \cref{thm:feedback_union} implies the corollary.
\end{proof}

\subsubsection{Generation in the Limit with Finite Feedback.}

Similar to the noisy and lossy models of generation which gave the adversary finite power,
it is natural to ask whether allowing a finite number of queries increases the power of a language generation algorithm.
This new model is similar to the previous model of generation with feedback.
However, each query $y_t$ can now either be a string, or the $\bot$ symbol,
which signifies that the algorithm is not querying a string at time $t$.
If $y_t = \bot$, then the response $a_t$ is also $\bot$.
Otherwise, we have $a_t \in \{\text{Yes}, \text{No}\}$ as before.
During the course of the execution, there may be at most a finite number of times $t$ where $y_t \neq \bot$.

\begin{definition}[Generator algorithm with $i$ queries]
    A generator algorithm with $i$ queries is a function $G$ that takes as input an alternating sequence
    of strings $x_t$ and responses $a_t$, and generates either a query string $y_t$
    or an output string $z_t$.
    If the input sequence ends with an enumerated string $x_t$,
    then the generator output $G(x_0, a_0, \dots, a_{t-1}, x_t)$
    represents the query string $y_t$.
    Otherwise, the generator output $G(x_0, a_0, \dots, x_t, a_t)$
    represents the output string $z_t$.
    Furthermore, for any infinite sequence $x_0$, $a_0$, $x_1$, \dots,
    there may be at most $i$ values of $t$ such that $G(x_0, a_0, \dots, a_{t-1}, x_t) \neq \bot$.
\end{definition}

\begin{definition}[Generation in the limit with $i$ queries]
    A generator algorithm $G$ with $i$ queries generates in the limit with $i$ queries
    for a collection $\SC$ if for any $K \in \SC$ and any enumeration $x$ of $K$,
    there exists a time step $\tst$ where for all $t \ge \tst$,
    the generated string $z_t$ at time $t$ is in $K \setminus S_t$.
\end{definition}

We now prove the second part of \Cref{thm:feedback_combined}, that having a finite number of queries does not give an
algorithm additional power for generation in the limit.

\begin{theorem}
    \label{thm:feedback_finite}
    For any collection $\SC$ and $i \in \NN$, if $\SC$ can be generated in the limit
    with $i$ queries, then $\SC$ can be generated in the limit without queries.
\end{theorem}

\begin{algorithm2e}
    \SetKwInput{Input}{Input}
    \Input{A generator algorithm $G$ with $i$ queries}

    $S = \emptyset$ \\
    \For{$t = 0, 1, 2, \dots$}{
        \label{line:query_for}
        Adversary reveals $x_t$ \\
        $S = S \cup \{x_t\}$ \\
        \For{$j = 0, \dots, t$}{
            $\seq{y}{t}{j} = G(x_0, \seq{a}{t}{0}, \dots, \seq{a}{t}{j-1}, x_j)$ \\
            \If{$\seq{y}{t}{j} = \bot$}{
                $\seq{a}{t}{j} = \bot$
            }
            \ElseIf{$\seq{y}{t}{j} \in S$}{
                $\seq{a}{t}{j} = \text{Yes}$
            }
            \Else{
                $\seq{a}{t}{j} = \text{No}$
            }
        }
        $z_t = G(x_0, \seq{a}{t}{0}, \dots, x_t, \seq{a}{t}{t})$ \\
        output $z_t$
    }
    \caption{Simulating an Algorithm with Finite Queries}
    \label{alg:query}
\end{algorithm2e}

\begin{proof}
    For any collection $\SC$ and $i \in \NN$, let $G$ be an algorithm which generates in the limit with $i$ queries.
    We claim that running \cref{alg:query} with input $G$
    produces an algorithm that generates in the limit for $\SC$ without using any queries.
    Intuitively, at each iteration of the for loop on line~\ref{line:query_for},
    we restart the simulation and answer each of the queries $\seq{y}{t}{j}$ by
    evaluating whether the query is in the current set of enumerated strings $S_t$.
    Eventually, every string in the target language must appear in $S_t$,
    so we will eventually correctly answer the queries and thus generate correctly.

    More formally, let $K \in \SC$ be an arbitrary language and $x$ be an arbitrary enumeration of $K$.
    We introduce the notion of a decision tree, which encodes the results of the queries
    during an iteration.
    For a given iteration of $t$ in the loop on line~\ref{line:query_for},
    let $Q_t = \{j \mid \seq{y}{t}{j} \neq \bot\}$ be the times at which the algorithm asked a query.
    Also let $d_t = \abs{Q_t}$ denote the total number of queries asked and $\seq{q}{t}{0}$, \dots, $\seq{q}{t}{d_t-1}$
    be the sorted list of times when the queries were asked.
    Finally, the sequence $\seq{s}{t}{0}$, \dots, $\seq{s}{t}{d_t-1}$ where $\seq{s}{t}{j} = \seq{y}{t}{\seq{q}{t}{j}}$
    represents the queries in order,
    and the sequence $\seq{r}{t}{0}$, \dots, $\seq{r}{t}{d_t-1}$ where $\seq{r}{t}{j} = \seq{a}{t}{\seq{q}{t}{j}}$
    represents the ordered responses to those queries.
    Note that $\seq{r}{t}{j} \in \{\text{Yes}, \text{No}\}$ for each $\seq{r}{t}{j}$.

    Now imagine a full binary tree of depth $i$, where for each non-leaf node,
    the left edge is labeled ``Yes'', and the right edge is labeled ``No''.
    We map a sequence $\seq{r}{t}{0}$, \dots, $\seq{r}{t}{d_t-1}$ to a node in the binary tree
    by starting at the root and then following the edge labeled by $\seq{r}{t}{j}$ when at depth $j$.
    Let $v_t$ be node mapped to by $\seq{r}{t}{0}$, \dots, $\seq{r}{t}{d_t-1}$.
    Now consider the preorder traversal of the binary tree
    where the traversal time of each node is the time at which it is first visited in a DFS.
    We claim that if $v_t \neq v_{t+1}$, then
    $v_t$ must appear before $v_{t+1}$ in the preorder traversal of the tree.
    First, if $d_{t+1} \ge d_t$ and $\seq{r}{t}{j} = \seq{r}{t+1}{j}$ for each $j < d_t$,
    then $v_{t+1}$ is a descendant of $v_{t}$, so the condition is satisfied.

    Next, we claim that it is not possible that $d_{t+1} < d_t$
    and $\seq{r}{t}{j} = \seq{r}{t+1}{j}$ for each $j < d_{t+1}$, i.e.,
    that $v_{t+1}$ is an ancestor of $v_t$.
    Note that in this case, the values of $\seq{a}{t}{i}$ and $\seq{a}{t+1}{i}$
    must be identical up to time $\seq{q}{t}{d_{t+1}}$.
    If $d_{t+1} < d_t$, then it must be that the query at time $\seq{q}{t}{d_{t+1}}$
    was not $\perp$ during iteration $t$, but is now $\perp$ during iteration $t+1$.
    However, the query at time $\seq{q}{t}{d_{t+1}}$ is solely a function of
    the enumerated strings $x$ and responses $a$ up until that time.
    Since the strings and responses up until time $\seq{q}{t}{d_{t+1}}$
    have not changed between iteration $t$ and $t+1$,
    the query at time $\seq{q}{t}{d_{t+1}}$ must also be the same between the two iterations.
    Thus, it cannot be that $d_{t+1} < d_t$
    and $\seq{r}{t}{j} = \seq{r}{t+1}{j}$ for each $j < d_{t+1}$.

    In the remaining case, there must be some first index $k < d_t$
    such that $\seq{r}{t}{k} \neq \seq{r}{t+1}{k}$.
    Since $k$ is the first index at which the responses differ, the outputs and queries up until time
    $\seq{q}{t}{k}$ must be the same between iterations $t$ and $t+1$.
    In particular, the queries $\seq{s}{t}{k}$ and $\seq{s}{t+1}{k}$ at time $\seq{q}{t}{k}$ must be equal.
    Thus, the only way for $\seq{r}{t}{k}$ to be different from $\seq{r}{t+1}{k}$
    is if the queried element $\seq{s}{t}{k}$ was newly added to $S$ during iteration $t+1$.
    In such a case, we must have that $\seq{r}{t}{k} = \text{No}$ and $\seq{r}{t+1}{k} = \text{Yes}$.
    Since ``Yes'' corresponds to the right edge, and the entire left subtree of a node
    appears earlier in the preorder traversal than the right subtree, we have that $v_t$
    appears before $v_{t+1}$ as desired.

    Since the binary tree has a finite number of nodes, and the preorder time of nodes
    weakly increases each iteration, there are only a finite number of iterations $t$
    where $v_{t} \neq v_{t+1}$.
    Thus, there must be some iteration $c < \infty$ at which point
    $v_t = v_c$ for all $t \ge c$.
    This also implies that both the sequence of queries $\seq{s}{t}{0}$, \dots, $\seq{s}{t}{d_t-1}$
    and the sequence of answers $\seq{r}{t}{0}$, \dots, $\seq{r}{t}{d_t-1}$
    are identical for every $t \ge c$.
    Note that every string $x \in K$ must eventually appear in the enumeration $x$.
    In particular, for every query string $y$, if in fact $y \in K$,
    then there must be some future iteration $t$ at which $y \in S$.
    Thus, if the queries and answers remain unchanged after iteration $c$,
    the answers to the queries must have all been correct at iteration $c$.

    To conclude, let $t_1$ be the time at which $G$ generates in the limit with $i$ queries
    for the language $K$ and enumeration $x$.
    We claim that after time $\tst = \max(c, t_1)$, the output of \Cref{alg:sim_without}
    must be correct unseen strings from $K$.
    First, for all times $t \ge \tst$, we showed that every query
    is answered correctly according to the actual target language $K$.
    Thus for all $t \ge \tst$, the output of \Cref{alg:sim_without} must coincide
    with the output of the algorithm $G$.
    Since $G$ generates correctly for all $t \ge \tst$, this proves the claim.
\end{proof}

\subsection{Non-Uniform Identification with Feedback.}

\begin{algorithm2e}
    \SetKwInput{Input}{Input}
    \Input{A countable collection $\SC = \{L_0, L_1, \dots\}$}

    $P = \emptyset$ \\
    $N = \emptyset$ \\
    \For{$t = 0, 1, 2, \dots$}{
        Adversary reveals $x_t$ \\
        $P = P \cup \{x_t\}$ \\
        Algorithm queries $t$ \\
        Adversary responds $a_t$ \\
        \If{$a_t = \text{Yes}$}{
            $P = P \cup \{t\}$ \\
        }\Else{
            $N = N \cup \{t\}$
        }
        $z_t = 0$ \\
        \For{$i = 0, \dots, t$}{
            \label{line:ident_for_i}
            \If{$P \subseteq L_i$ and $N \cap L_i = \emptyset$}{
                \label{line:ident_if}
                $z_t = i$ \\
                \textbf{break}
            }
        }
        output $z_t$
    }
    \caption{Non-uniform identifier with feedback}
    \label{alg:identify}
\end{algorithm2e}

Finally, we consider a model of identification in the limit with feedback for \emph{countable} collections
which was not discussed in the introduction.
In this model, we assume that the languages in the collection are explicitly indexed
so that $\SC = \{L_0, L_1, \dots\}$.
As before, the adversary selects an arbitrary language $K \in \SC$ and an arbitrary
enumeration of the language.
Once again, the algorithm can query at each time step $t$ whether a string $y_t$ is in $K$.
However, the key difference is that rather outputting a string,
the output $z_t$ is an integer, which represents a guess that the target language is equal to $L_{z_t}$.

\begin{definition}[Identifier algorithm with feedback]
    An identifier algorithm with feedback is a function $G$ that takes as input an alternating sequence
    of strings $x_t$ and responses $a_t$, and generates either a query string $y_t$
    or an output index $z_t$.
    If the input sequence ends with an enumerated string $x_t$,
    then the generator output $G(x_0, a_0, \dots, a_{t-1}, x_t)$
    represents the query string $y_t$.
    Otherwise, the generator output $G(x_0, a_0, \dots, x_t, a_t)$
    represents the output index $z_t$.
\end{definition}

\begin{definition}[Non-uniform identification with feedback]
    An identifier algorithm with feedback $G$ non-uniformly identifies with feedback for a collection $\SC$
    if for any $K \in \SC$, there exists a time step $\tst$
    such that for any enumeration $x$ of $K$, we have $K = L_{z_t}$ for all $t \ge \tst$.
\end{definition}

We give an algorithm that non-uniformly identifies with feedback for every countable collection.

\identfeed*
\begin{proof}
    Let $\SC = \{L_0, L_1, \dots\}$ be a countable collection
    and assume without loss of generality that $\SC$ is a collection over $\NN$.
    We claim that \Cref{alg:identify} on input $\SC$ non-uniformly identifies with feedback.
    Fix an arbitrary $K \in \SC$ and let $k$ be the first index such that $L_k = K$.
    For every $j < k$, there must be a smallest integer $t_j$ such that $t_j$
    appears in exactly one of $L_j$ or $L_k$.
    We claim that after time $\tst = \max(k, t_0, \dots, t_{k-1})$,
    \Cref{alg:identify} will correctly output the index $k$.

    Since $\tst \ge k$, the language $L_k$ is under consideration
    in the for loop on line~\ref{line:ident_for_i} at time $\tst$.
    Now note that at time $\tst$, the algorithm has queried every integer from $0$ to $\tst$.
    For every $j < k$, since $t_j$ appears in exactly one of $L_j$ and $L_k$,
    it must be that either $t_j \in P$ and $t_j \notin L_j$, or $t_j \in N$ and $t_j \in L_j$.
    Thus it must be true that either $P \setminus L_i \neq \emptyset$ or $N \cap L_i \neq \emptyset$.
    In either case, the if condition on line~\ref{line:ident_if} is false for all $j < k$.
    Clearly, the if condition is true when $i=k$.
    Thus, we correctly output the index $k$ as desired.
\end{proof}

    \section{Closing Remarks.}\label{sec:closing}Language generation in the limit is an exciting new formalism for understanding the fundamental structure and limitations of language learning. Since the work of Kleinberg and Mullainathan~\cite{KM24}, in less than a year, a wealth of research has addressed different aspects of this phenomenon, and introduced new refinements and variations to the basic model. In this paper, we answer one of the basic open questions posed in this line of research by Li, Raman, and Tewari~\cite{LRT25} about the union-closedness of language generation in the limit. We then give precise characterizations of the qualitative and quantitative roles of three important ingredients---loss (omissions), noise (errors), and feedback (membership queries)---in the language generation problem, which were studied in previous work (\cite{LRT25,RR25,CP25}). We close the paper with the belief that language generation is a phenomenon of fundamental importance that will be studied extensively in the coming years, and hope that our results will help lay the groundwork for further theoretical explorations in this domain.

    \section*{Acknowledgments.}\label{sec:ack}

    This research was supported in part by NSF awards CCF-1955703 and CCF-2329230.
    DP would also like to thank Jon Kleinberg for an interesting talk and discussion
    on language generation in the limit during a visit to Duke University, that got him interested in the topic.


    \bibliographystyle{siamplain}
    \bibliography{main_siam}

    \appendix

    \section*{Appendix}

    \section{Equivalence of Prior Models of Language Generation.}\label{sec:appendix}We now discuss the slight differences between our definitions of generation
compared to the definitions found in~\cite{KM24, LRT25}.
In previous literature, the adversary has been allowed to repeat strings in its enumeration.
In contrast, we define enumerations to be infinite sequences of unique strings.
We now show that these definitions are in fact equivalent.
For clarity, we will refer to our previous definitions as generation without repetition.
We refer to models where the adversary may repeat strings as generation with repetition.

When the adversary may repeat strings, we can no longer hope for
a fixed time step at which the algorithm generates correctly,
since the adversary can repeatedly output a single string for an arbitrarily long period of time.
Instead, we must require the algorithm to generate successfully after the adversary outputs sufficiently
many distinct strings.
Indeed, the benefit of our definitions without repetition is that we may require a fixed time
after which the algorithm must generate correctly.

\begin{definition}[Uniform generation with repetitions~\cite{LRT25}]
    An algorithm $G$ uniformly generates with repetitions for a collection $\SC$
    if there exists a $\dst$ such that for any $K \in \SC$ and any sequence $x_0$, $x_1$,
    \dots\ with $S_{\infty} = K$, the generated string $z_t$ is in $K \setminus S_t$
    for all times $t$ such that $\abs{S_t} \ge \dst$.
\end{definition}

\begin{definition}[Non-uniform generation with repetitions~\cite{LRT25}]
    An algorithm $G$ non-uniformly generates with repetitions for a collection $\SC$
    if for any $K \in \SC$, there exists a $\dst$ such that for any sequence $x_0$, $x_1$,
    \dots\ with $S_{\infty} = K$, the generated string $z_t$ is in $K \setminus S_t$
    for all times $t$ such that $\abs{S_t} \ge \dst$.
\end{definition}

The definition for generation in the limit remains the same,
except that the adversary may repeat strings.

\begin{definition}[Generation in the limit with repetitions]
    An algorithm $G$ generates in the limit with repetitions for a collection $\SC$
    if for any $K \in \SC$, and any sequence $x_0$, $x_1$, \dots\ with $S_{\infty} = K$,
    there exists a $\tst$ such that the generated string $z_t$ is in $K \setminus S_t$
    for all times $t \ge \tst$.
\end{definition}

Clearly, for all three variants, if a collection $\SC$ is generatable with repetitions,
then $\SC$ must also be generatable without repetitions.
We now show for each variant that if $\SC$ is generatable without repetitions,
then $\SC$ must also be generatable with repetitions.

\begin{algorithm2e}
    \SetKwInput{Input}{Input}
    \Input{A generator $G$}

    Adversary reveals $x_0$ \\
    $y_0 = x_0$ \\
    output $z_0 = G(y_0)$ \\
    $S = \{x_0\}$ \\
    $t = 1$ \\
    \For{$i = 1, 2, \dots$}{
        Adversary reveals $x_t$ \\
        \While{$x_t \in S$}{
            output $z_t = G(y_0, \dots, y_{i-1})$ \\
            $t = t + 1$ \\
            Adversary reveals $x_t$
        }
        $S = S \cup \{x_t\}$ \\
        $y_i = x_t$ \\
        output $z_t = G(y_0, \dots, y_i)$ \\
        $t = t + 1$ \\
    }
    \caption{Generator with repetitions}
    \label{alg:repeat}
\end{algorithm2e}

\begin{lemma}
    If a collection $\SC$ is uniformly generatable without repetition,
    then $\SC$ is uniformly generatable with repetitions.
\end{lemma}
\begin{proof}
    Let $\SC$ be an arbitrary collection and let $G$ uniformly generate without repetitions for $\SC$.
    We claim that \Cref{alg:repeat} on input $G$ is a uniform generator with repetitions for $\SC$.
    Intuitively, \Cref{alg:repeat} takes the adversary enumeration $x$
    and generates a sequence $y_0$, $y_1$, \dots\ without repetitions by keeping the first
    occurrence of each string in $x$ and discarding the subsequent occurrences.
    Then at each step, we simply generate according to the sequence $y$.

    More formally, let $\tst$ be the time at which $G$ generates correctly without repetitions.
    Fix an arbitrary $K \in \SC$ and enumeration $x$ of $K$ with repetitions.
    Let $\dst$ be the first time at which $\abs{\{x_0, \dots, x_{\dst}\}} = \tst$.
    It is easy to see that at any time $t$ during the execution of the algorithm,
    we have $S = \{x_0, \dots, x_{t}\} = \{y_0, \dots, y_i\}$.
    Thus at any time $t \ge \dst$, the sequence $y_0$, \dots, $y_i$ must consist of
    at least $\tst$ distinct strings.
    Since $G$ must generate correctly at time $\tst$, the output $z_t$ is in $K \setminus S_t$ as desired.
\end{proof}

The same proof shows the corresponding result for non-uniform generation.

\begin{lemma}
    If a collection $\SC$ is non-uniformly generatable without repetition,
    then $\SC$ is non-uniformly generatable with repetitions.
\end{lemma}
\begin{proof}
    Let $\SC$ be an arbitrary collection and let $G$ non-uniformly generate without repetitions for $\SC$.
    We claim that \Cref{alg:repeat} on input $G$ is a non-uniform generator with repetitions for $\SC$.
    Fix an arbitrary $K \in \SC$ and let $\tst$ be the time at which
    $G$ generates correctly without repetitions for $K$.
    Fix an arbitrary enumeration $x$ of $K$ with repetitions
    and let $\dst$ be the first time at which $\abs{\{x_0, \dots, x_{\dst}\}} = \tst$.
    It is easy to see that at any time $t$ during the execution of the algorithm,
    we have $S = \{x_0, \dots, x_{t}\} = \{y_0, \dots, y_i\}$.
    Thus at any time $t \ge \dst$, the sequence $y_0$, \dots, $y_i$ must consist of
    at least $\tst$ distinct strings.
    Since $G$ must generate correctly at time $\tst$, the output $z_t$ is in $K \setminus S_t$ as desired.
\end{proof}

We conclude with the corresponding proof for generation in the limit.

\begin{lemma}
    If a collection $\SC$ is generatable in the limit without repetition,
    then $\SC$ is generatable in the limit generatable with repetitions.
\end{lemma}
\begin{proof}
    Let $\SC$ be an arbitrary collection and let $G$ non-uniformly generate without repetitions for $\SC$.
    We claim that \Cref{alg:repeat} on input $G$ is a non-uniform generator with repetitions for $\SC$.
    Fix an arbitrary $K \in \SC$ and an arbitrary enumeration $x$ of $K$ with repetitions.
    Now consider the sequence $y_0$, $y_1$, \dots\ generated by \Cref{alg:repeat}
    when the adversary reveals strings according to the enumeration $x$.
    It is clear that $y$ is an enumeration of $K$ without repetitions.
    Thus let $\tst$ be the time at which $G$ generates correctly given the enumeration $y$.
    It is easy to see that at any time $t$ during the execution of the algorithm,
    we have $S = \{x_0, \dots, x_{t}\} = \{y_0, \dots, y_i\}$.
    Thus at any time $t \ge \dst$, the sequence $y_0$, \dots, $y_i$ must consist of
    at least $\tst$ distinct strings.
    Since $G$ must generate correctly at time $\tst$, the output $z_t$ is in $K \setminus S_t$ as desired.
\end{proof}


\end{document}